\documentclass{lmcs} 
\pdfoutput=1

\usepackage{lastpage}
\lmcsdoi{17}{3}{15}
\lmcsheading{}{\pageref{LastPage}}{}{}%
{Feb.~10,~2020}{Jul.~30,~2021}{}

\usepackage[utf8]{inputenc}


\usepackage{tikz}
\usetikzlibrary{automata,positioning,cd}
\usepackage{wrapfig}

\usepackage{todonotes}
\usepackage{amsmath}
\usepackage{amssymb}
\usepackage{amstext}
\usepackage{mathtools}

\usepackage{xypic}
\usepackage[all,cmtip]{xy}
\xyoption{2cell}
\UseTwocells
\xyoption{tips}

\newcommand{\weg}[1]{}

\theoremstyle{definition}\newtheorem{convention}[thm]{Convention}

\newcommand{\Set}{\mathbf{Set}}
\newcommand{\Rel}{\mathbf{Rel}}
\newcommand{\op}[1]{#1^{\mathrm{op}}}
\newcommand{\fop}[1]{#1^{\mathrm{fop}}}

\ifdefined\mathsl\else
  \DeclareMathAlphabet{\mathsl}{\encodingdefault}{\rmdefault}{\mddefault}{\sldefault}
  \SetMathAlphabet{\mathsl}{bold}{\encodingdefault}{\rmdefault}{\bfdefault}{\sldefault}
\fi

\newcommand{\CoAlg}{\mathsl{CoAlg}}
\newcommand{\rellift}{\mathsl{Rel}}
\newcommand{\foprellift}{\fop{\mathsl{Rel}}}
\newcommand{\set}[2]{\{#1\;|\;#2\}}

\newcommand{\all}[2]{\forall#1.\,#2}
\newcommand{\allin}[3]{\forall#1\in#2.\,#3}
\newcommand{\ex}[2]{\exists#1.\,#2}
\newcommand{\exin}[3]{\exists#1\in#2.\,#3}

\newcommand{\tuple}[1]{\langle#1\rangle}
\newcommand{\idmap}[1][]{\ensuremath{\mathrm{id}_{#1}}}
\newcommand{\after}{\mathrel{\circ}}
\newcommand{\Pow}{\mathcal{P}}

\newcommand{\Ty}{\mathrm{Ty}}
\newcommand{\At}{\mathrm{At}}
\newcommand{\Eq}{\mathsl{Eq}}
\newcommand{\nEq}{\mathsl{nEq}}
\newcommand{\inl}{\mathrm{inl}}
\newcommand{\inr}{\mathrm{inr}}
\newcommand{\pair}[2]{\langle{#1},{#2}\rangle}

\newcommand{\bisim}[2]{\stackrel{#1}{\leftrightarrow}_{#2}}
\newcommand{\apart}[2]{\stackrel{#1}{\#}_{#2}}
\newcommand{\bis}{\mathrel{\underline{\leftrightarrow}}}
\newcommand{\apt}{\mathrel{\underline{\#}}}
\newcommand{\branbis}{\mathrel{\underline{\leftrightarrow}_b}}
\newcommand{\sbranbis}{\mathrel{\underline{\leftrightarrow}_{sb}}}
\newcommand{\rbranbis}{\mathrel{\underline{\leftrightarrow}_{rb}}}
\newcommand{\weakbis}{\mathrel{\underline{\leftrightarrow}_w}}
\newcommand{\branap}{\mathrel{\apt_b}}
\newcommand{\branapA}{\mathrel{\apt^A_b}}
\newcommand{\sbranap}{\mathrel{\apt_{sb}}}
\newcommand{\rbranap}{\mathrel{\apt_{rb}}}

\newcommand{\weakap}{\mathrel{\apt_w}}
\newcommand{\trans}{\rightarrow}
\newcommand{\transa}{\trans_a}
\newcommand{\transx}{\trans_x}
\newcommand{\transu}{\trans_u}
\newcommand{\transt}{\trans_{\tau}}
\newcommand{\transtt}{\twoheadrightarrow_{\tau}}
\newcommand{\Atau}{A_{\tau}}
\newcommand{\symm}{\mathrm{symm}}
\newcommand{\weakbisa}{\mathrm{bis}_{w}}
\newcommand{\weakbist}{\mathrm{bis}_{w\tau}} 
\newcommand{\branbisa}{\mathrm{bis}_{b}}
\newcommand{\branbist}{\mathrm{bis}_{b\tau}} 
\newcommand{\sbranbisa}{\mathrm{bis}_{sb}}
\newcommand{\sbranbist}{\mathrm{bis}_{sb\tau}}
\newcommand{\rbranbisx}{\mathrm{bis}_{rb}}
\newcommand{\weakapa}{\mathrm{in}_{w}}
\newcommand{\weakapt}{\mathrm{in}_{w\tau}} 
\newcommand{\branapa}{\mathrm{in}_{b}}
\newcommand{\rbranapx}{\mathrm{in}_{rb}}

\newcommand{\branapaalt}{\mathrm{in}^A_{b}}
\newcommand{\sbranapt}{\mathrm{in}_{sb\tau}}
\newcommand{\branapt}{\mathrm{in}_{b\tau}} 
\newcommand{\stut}{\mathrm{stut}}

\newcommand{\QEDbox}{\square}
\newcommand{\QED}{\hspace*{\fill}$\QEDbox$}

\newdimen\proofrulebreadth \proofrulebreadth=.05em
\newdimen\proofdotseparation \proofdotseparation=1.25ex
\newdimen\proofrulebaseline \proofrulebaseline=2ex
\newcount\proofdotnumber \proofdotnumber=3
\let\then\relax
\def\hfi{\hskip0pt plus.0001fil}
\mathchardef\squigto="3A3B
\newif\ifinsideprooftree\insideprooftreefalse
\newif\ifonleftofproofrule\onleftofproofrulefalse
\newif\ifproofdots\proofdotsfalse
\newif\ifdoubleproof\doubleprooffalse
\let\wereinproofbit\relax
\newdimen\shortenproofleft
\newdimen\shortenproofright
\newdimen\proofbelowshift
\newbox\proofabove
\newbox\proofbelow
\newbox\proofrulename
\def\shiftproofbelow{\let\next\relax\afterassignment\setshiftproofbelow\dimen0 
}
\def\shiftproofbelowneg{\def\next{\multiply\dimen0 by-1 }%
\afterassignment\setshiftproofbelow\dimen0 }
\def\setshiftproofbelow{\next\proofbelowshift=\dimen0 }
\def\setproofrulebreadth{\proofrulebreadth}

\def\prooftree{
\ifnum  \lastpenalty=1
\then   \unpenalty
\else   \onleftofproofrulefalse
\fi
\ifonleftofproofrule
\else   \ifinsideprooftree
        \then   \hskip.5em plus1fil
        \fi
\fi
\bgroup
\setbox\proofbelow=\hbox{}\setbox\proofrulename=\hbox{}%
\let\justifies\proofover\let\leadsto\proofoverdots\let\Justifies\proofoverdbl
\let\using\proofusing\let\[\prooftree
\ifinsideprooftree\let\]\endprooftree\fi
\proofdotsfalse\doubleprooffalse
\let\thickness\setproofrulebreadth
\let\shiftright\shiftproofbelow \let\shift\shiftproofbelow
\let\shiftleft\shiftproofbelowneg
\let\ifwasinsideprooftree\ifinsideprooftree
\insideprooftreetrue
\setbox\proofabove=\hbox\bgroup$\displaystyle 
\let\wereinproofbit\prooftree
\shortenproofleft=0pt \shortenproofright=0pt \proofbelowshift=0pt
\onleftofproofruletrue\penalty1
}

\def\eproofbit{
\ifx    \wereinproofbit\prooftree
\then   \ifcase \lastpenalty
        \then   \shortenproofright=0pt  
        \or     \unpenalty\hfil         
        \or     \unpenalty\unskip       
        \else   \shortenproofright=0pt  
        \fi
\fi
\global\dimen0=\shortenproofleft
\global\dimen1=\shortenproofright
\global\dimen2=\proofrulebreadth
\global\dimen3=\proofbelowshift
\global\dimen4=\proofdotseparation
\global\count255=\proofdotnumber
$\egroup  
\shortenproofleft=\dimen0
\shortenproofright=\dimen1
\proofrulebreadth=\dimen2
\proofbelowshift=\dimen3
\proofdotseparation=\dimen4
\proofdotnumber=\count255
}

\def\proofover{
\eproofbit 
\setbox\proofbelow=\hbox\bgroup 
\let\wereinproofbit\proofover
$\displaystyle
}%
\def\proofoverdbl{
\eproofbit 
\doubleprooftrue
\setbox\proofbelow=\hbox\bgroup 
\let\wereinproofbit\proofoverdbl
$\displaystyle
}%
\def\proofoverdots{
\eproofbit 
\proofdotstrue
\setbox\proofbelow=\hbox\bgroup 
\let\wereinproofbit\proofoverdots
$\displaystyle
}%
\def\proofusing{
\eproofbit 
\setbox\proofrulename=\hbox\bgroup 
\let\wereinproofbit\proofusing
\kern0.3em$
}

\def\endprooftree{
\eproofbit 
  \dimen5 =0pt
\dimen0=\wd\proofabove \advance\dimen0-\shortenproofleft
\advance\dimen0-\shortenproofright
\dimen1=.5\dimen0 \advance\dimen1-.5\wd\proofbelow
\dimen4=\dimen1
\advance\dimen1\proofbelowshift \advance\dimen4-\proofbelowshift
\ifdim  \dimen1<0pt
\then   \advance\shortenproofleft\dimen1
        \advance\dimen0-\dimen1
        \dimen1=0pt
        \ifdim  \shortenproofleft<0pt
        \then   \setbox\proofabove=\hbox{%
                        \kern-\shortenproofleft\unhbox\proofabove}%
                \shortenproofleft=0pt
        \fi
\fi
\ifdim  \dimen4<0pt
\then   \advance\shortenproofright\dimen4
        \advance\dimen0-\dimen4
        \dimen4=0pt
\fi
\ifdim  \shortenproofright<\wd\proofrulename
\then   \shortenproofright=\wd\proofrulename
\fi
\dimen2=\shortenproofleft \advance\dimen2 by\dimen1
\dimen3=\shortenproofright\advance\dimen3 by\dimen4
\ifproofdots
\then
        \dimen6=\shortenproofleft \advance\dimen6 .5\dimen0
        \setbox1=\vbox to\proofdotseparation{\vss\hbox{$\cdot$}\vss}%
        \setbox0=\hbox{%
                \advance\dimen6-.5\wd1
                \kern\dimen6
                $\vcenter to\proofdotnumber\proofdotseparation
                        {\leaders\box1\vfill}$%
                \unhbox\proofrulename}%
\else   \dimen6=\fontdimen22\the\textfont2 
        \dimen7=\dimen6
        \advance\dimen6by.5\proofrulebreadth
        \advance\dimen7by-.5\proofrulebreadth
        \setbox0=\hbox{%
                \kern\shortenproofleft
                \ifdoubleproof
                \then   \hbox to\dimen0{%
                        $\mathsurround0pt\mathord=\mkern-6mu%
                        \cleaders\hbox{$\mkern-2mu=\mkern-2mu$}\hfill
                        \mkern-6mu\mathord=$}%
                \else   \vrule height\dimen6 depth-\dimen7 width\dimen0
                \fi
                \unhbox\proofrulename}%
        \ht0=\dimen6 \dp0=-\dimen7
\fi
\let\doll\relax
\ifwasinsideprooftree
\then   \let\VBOX\vbox
\else   \ifmmode\else$\let\doll=$\fi
        \let\VBOX\vcenter
\fi
\VBOX   {\baselineskip\proofrulebaseline \lineskip.2ex
        \expandafter\lineskiplimit\ifproofdots0ex\else-0.6ex\fi
        \hbox   spread\dimen5   {\hfi\unhbox\proofabove\hfi}%
        \hbox{\box0}%
        \hbox   {\kern\dimen2 \box\proofbelow}}\doll%
\global\dimen2=\dimen2
\global\dimen3=\dimen3
\egroup 
\ifonleftofproofrule
\then   \shortenproofleft=\dimen2
\fi
\shortenproofright=\dimen3
\onleftofproofrulefalse
\ifinsideprooftree
\then   \hskip.5em plus 1fil \penalty2
\fi
}

\newenvironment{myproof}{\begin{trivlist} \item[\hskip \labelsep%
{\bf Proof.}]}{\end{trivlist}}

\begin{document}

\title[Apartness and Bisimulation]{Relating Apartness and Bisimulation}

\author[H.~Geuvers]{Herman Geuvers\rsuper{{a,b}}}	
\address{\lsuper{a}ICIS, Radboud University Nijmegen}
\address{\lsuper{b}Faculty of Mathematics and Computer Science, Technical University Eindhoven}	
\author[B.~Jacobs]{Bart Jacobs\rsuper{a}}	
\email{herman@cs.ru.nl}  
\email{bart@cs.ru.nl}  





\begin{abstract}
A bisimulation for a coalgebra of a functor on the category of sets
can be described via a coalgebra in the category of relations, of a
lifted functor. A final coalgebra then gives rise to the coinduction
principle, which states that two bisimilar elements are equal. For
polynomial functors, this leads to well-known descriptions. In the
present paper we look at the dual notion of
``apartness''. Intuitively, two elements are apart if there is a
positive way to distinguish them. Phrased differently: two elements
are apart if and only if they are not bisimilar. Since apartness is an
inductive notion, described by a least fixed point, we can give a
proof system, to derive that two elements are apart. This proof system
has derivation rules and two elements are apart if and only if there
is a finite derivation (using the rules) of this fact.

We study apartness versus bisimulation in two separate ways. First,
for weak forms of bisimulation on labelled transition systems, where
silent ($\tau$) steps are included, we define an apartness notion that
corresponds to weak bisimulation and another apartness that
corresponds to branching bisimulation. The rules for apartness can be
used to show that two states of a labelled transition system are not
branching bismilar. To support the apartness view on labelled
transition systems, we cast a number of well-known properties of
branching bisimulation in terms of branching apartness and prove
them. Next, we also study the more general categorical situation and
show that indeed, apartness is the dual of bisimilarity in a precise
categorical sense: apartness is an initial algebra and gives rise to
an induction principle. In this analogy, we include the powerset
functor, which gives a semantics to non-deterministic choice in
process-theory.
\end{abstract}

\maketitle

\section{Introduction}\label{S:one}

Bisimulation is a standard way of looking at indistinguishability of
processes, labelled transitions, automata and streams, etc. These
structures all have in common that they can be seen as coalgebraic:
the elements are not built inductively, using constructors, but they
are observed through ``destructors'' or ``transition maps''. The
coinduction principle states that two elements that have the same
observations are equal, when mapped to a ``final'' model. A
bisimulation is a relation that is preserved along transitions: if two
elements are bisimilar, and we perform a transition, then we either
get two new bisimilar elements, or we get equal outputs (in case our
observation is a basic value). Two elements are bisimilar if and only if they are
observationally indistinguishable, that is, if there is a bisimulation
that relates them.

Coalgebraic structures have a natural notion of bisimulation, because
the transfer principle can be defined directly from the type of the
destructor, that is, from the functor involved. So bisimilarity, being
the largest bisimulation is also defined directly from the destructor
(transition operation), and it is well known that if one starts from a
final coalgebra, then bisimilarity on the final coalgebra coincides
with equality. This gives the coinduction principle: bisimilarity
implies equality, see e.g.~\cite{JacobsR11,Rutten00,Jacobs16}.

There is a dual way of looking at this, which has not been explored
much\footnote{One of the authors (BJ) did write an article about
  bisimulation and apartness in 1995, entitled \emph{Bisimulation and
    Apartness in Coalgebraic Specification}; it is available online at
  \href{http://citeseerx.ist.psu.edu/viewdoc/summary?doi=10.1.1.50.4507}{citeseerx.ist.psu.edu/viewdoc/summary?doi=10.1.1.50.4507}
  but was never published. Parts of that article are incorporated in
  the present text, esp.~in Section~\ref{ApartnesSec}.}. Of course,
the concept of observations is well-known, and there is work by Korver
\cite{Korver}, who presents an algorithm that, if two states are not
branching bisimilar, produces a formula in Hennessy-Milner
\cite{HennessyMilner} logic with until operator that distinguishes the
two states. Another work is Chow~\cite{Chow} on testing equivalence of
states in finite state machines and more recent work is by Smetsers et
al.~\cite{SmetsersMoermanJansen}, where an efficient algorithm is
presented for finding a minimal separating sequence for a pair of
in-equivalent states in a finite state machine.

We take this further by developing the basic parts of a theory of
``apartness''. The idea is that two elements are apart if we can make
an observation in finitely many steps that distinguishes these
elements. This idea goes back to Brouwer, in his approach to real
numbers, but here we introduce the notion of an ``apartness relation''
for a coalgebra, again directly from the definition of the type of the
destructor, i.e.~from the functor. Basically, a relation is an
apartness relation if it satisfies the inverse of the transfer
principle for bisimulations. We define two elements to be apart if
they are in all apartness relations. It can be shown that a relation
$Q$ is an apartness relation if and only if its complement $\neg Q$ is
a bisimulation relation. Thereby, two elements are apart if and only
if they are not bisimilar, that is, distinguishable. Aside from
providing a new view on bisimulation, apartness---being an inductive
notion---also provides a proof system: two elements are apart if and
only if there is a (finite, well-founded) derivation of that fact
using the derivation rules. These derivation rules are the rules that
define what an apartness is for that particular coalgebra, so they are
directly derived from the type of the destructor. This again
emphasizes that bisimilarity of two elements has to be proven
coinductively, while ``being apart'' for two elements can be proven by
giving a finite derivation (inductively).

This paper consists of two separate parts, one more concrete and one
more abstract.  The first, concrete part focuses on bisimulation and
apartness for labelled transition systems (LTS). We apply these
notions to the case of weak forms of bisimulation for labelled
transition systems with ``silent steps'', usually referred to as
$\tau$-steps. Silent steps cannot be directly observed, but sometimes
they do have some implicit side-effects as they may move the system
from a state where a certain action is enabled to a state where this
action is impossible. Therefore, several variations have been defined,
like weak bisimulation and branching bisimulation. We study these from
the point of view of apartness, and we define what it means to be a
``weak apartness'' relation and a ``branching apartness''
relation. Two states in a system are ``weakly apart'' if they are in
the intersection of all weak apartness relations and are ``branching
apart'' if they are in the intersection of all branching apartness
relations. The main outcome of this first part is a derivation system
for branching apartness. This is a derivation system in the
traditional (inductive) sense: a judgment holds if there is a finite
derivation (so no infinite or circular derivations)
that has that judgment as its conclusion. To show that the apartness
view on LTSs is fruitful, we use the derivation system for branching
apartness to show that the branching apartness relation is {\em
  co-transitive\/} and satisfies the {\em apartness stuttering
  property}. (These notions will be dealt with in Section
\ref{sec:apartbisim}.) These imply the stuttering property and the
transitivity for branching bisimulation, properties that are known to
be subtle to prove. (See \cite{Basten}.)  We also indicate how the
derivation system can be used as an algorithm for proving branching
apartness of two states in an LTS and we define and discuss the notion
of {\em rooted branching apartness\/} which is the dual of rooted
branching bisimulation.

The second part switches to a more abstract categorical level. It is
restricted however to functors on the category of sets. First, the
standard coalgebraic approach is recalled, in which a bisimulation is
a coalgebra itself, for a lifting of the functor involved to the
category of relations. This can be applied in particular to polynomial
functors and yields familiar descriptions of bisimulation.

Next, apartness is described in an analogous manner. It does not use
the category $\Rel$ of relations, nor its usual opposite $\op{\Rel}$,
but a special ``fibred'' opposite $\fop{\Rel}$. A special lifting of a
functor to $\fop{\Rel}$ is described, via negation as a functor $\neg
\colon \Rel \rightarrow \fop{\Rel}$. An apartness relation is then
defined as a coalgebra of the lifted functor (to $\fop{\Rel}$).  This
set-up then guarantees that a relation $R$ is a bisimulation iff $\neg
R$ is an apartness relation. Moreover, there is an analogue of the
coinduction principle, stating that two states of a coalgebraic system
are apart iff they are non-equal when mapped to the final coalgebra.

A significant conclusion from this analysis is: bisimilarity is the
\emph{greatest} fixed point in a partial order of relations.  But
apartness is the \emph{least} fixed point in that order.  This means
that apartness can be established in a finite number of steps. Hence
it can be described via a system of proof rules. This, in the end, is
the main reason why apartness can be more amenable than bisimulation.

We should emphasize that the two parts of this paper are really
``apart'' since there is no overlap. There is quite a bit of work on
dealing with weak/branching bisimulation in a coalgebraic setting (see
e.g.~\cite{SokolovaVinkWoracek,BeoharKupper,Brengos15,BrengosMP15,GoncharovP14}),
but there is no generic, broadly applicable approach. In this paper we
are not solving this longstanding open problem. We have separate
descriptions of weak/branching apartness (in the first part) and of
categorical apartness (in the second part). The only hope that we can
offer at this stage is that apartness might provide a fresh
perspective on a common approach.

To clarify some terminology and relate the corresponding notions of
the bisimulation view and the apartness view, we give the following
table.

\begin{center}
\begin{tabular}{|l|l|}
  \hline
bisimulation relation     & apartness relation\\
coinductive               & inductive\\
bisimulation equivalence  & proper apartness\\
congruence                & strong extensionality\\
\hline
\end{tabular}
\end{center}

A bisimulation relation models an equality of processes or process
terms, whereas an apartness relation models an inequality, so $R$ will
be a bisimulation (of some type) if and only if $\neg R$ is an
apartness (of that same type). Bisimilarity is the largest
bisimulation relation, which means that it is a coinductively defined
concept. Apartness is the smallest apartness relation, which means
that it is an inductively defined concept. A bisimulation should be
(at least) an equivalence relation, meaning that it satisfies
reflexivity, symmetry and transitivity. The dual notions are
irreflexivity, symmetry and co-transitivity, which together are
usually called ``apartness'' in the literature. To avoid confusion, we
have introduced the terminology ``proper apartness'' for a relation
that satisfies irreflexivity, symmetry and co-transitivity. In process
theory, bisimulation is not an equivalence relation by definition, so
neither is an apartness a ``proper apartness'' by definition. There is
really some work to do, so therefore it is important to single out
these notions. A relation $R$ is a congruence in case it is preserved
by application of operators: if $R(x,y)$, then $R((f(x), f(y))$ for
any operator $f$. The dual notion is strong extensionality, but in the
``apartness view'', this is not a property of the relation but of the
operator. Operator $f$ is strongly extensional (for apartness relation
$Q$) if $Q(f(x),f(y))$ implies $Q(x,y)$. (Intuitively: if $f(x)$ and
$f(y)$ are different, then $x$ and $y$ should be different.)

\subsection{Contents of the sections}
In Section \ref{sec:bisapt}, we introduce bisimulation and apartness
for streams and for deterministic automata, as preparation for more
general/complicated cases. In Section~\ref{sec:weakbran}, we discuss
weak and branching bisimulation and apartness and we indicate the
potential use of reasoning with apartness instead of bisimulation. In
Section~\ref{CoalgLiftSec} we recap the coalgebraic treatment of
bisimulation for coalgebras in the category $\Set$ as a coalgebra in
the category $\Rel$.  In Section~\ref{ApartnesSec} we introduce the
dual case and give a coalgebraic treatment of apartness, as the
opposite of bisimulation. For completeness, we give, in the Appendix,
a syntactic treatment of the general picture of
Section~\ref{sec:bisapt}, where we have a general type of coalgebras
for which we define bisimulation and apartness.

\subsection*{Special Thanks} We dedicate this article to Jos Baeten on
the occasion of his retirement. Much of Jos' research has centered
around process theory and process algebra, where various forms of
bisimulation equivalence have always played a central role. As a math
student, before going to the USA to do a PhD on a topic in the
intersection of recursion theory and set theory, Jos was part of the
Dutch ``school'' on constructive mathematics, and we think that it is
nice to see that `apartness', a notion which originates from
constructive mathematics, also has a natural place in the study of
process (non-)equivalence. The first author in particular would like
to thank Jos for the years he has worked at the Technical University
Eindhoven in the Formal Methods group, led by Jos, the many things he
has learned during this period and the pleasant cooperation on topics
of science, education and organisation. Thanks Jos!


\section{Bisimulation and apartness for streams and deterministic automata}
\label{sec:bisapt}
We start from the coalgebra of streams over an alphabet $A$ and the
coalgebra of DAs (Deterministic Automata) over $A$, for which we
illustrate the notions of bisimulation and apartness. We work in the
category $\Set$ of sets and functions. The coalgebra of streams over
$A$ is given by a function $c =\pair{h}{t} \colon K \rightarrow A
\times K$, where we associate every $s\in K$ with a stream by letting
$h(s) \in A$ denote the head of $s$ and $t(s) \in K$ the tail of $s$.

\begin{defi}\label{def.streams}
Let $A$ be a fixed set/alphabet. A coalgebraic map $\pair{h}{t} : K
\rightarrow A\times K$ gives rise to the following notions of {\em
  bisimulation for $c$\/} and {\em apartness for $c$}.
\begin{enumerate}
\item A relation $R\subseteq K\times K$ is a {\em $c$-bisimulation\/}
  if it satisfies the following rule
  \[\begin{prooftree}
    R(s_1, s_2)
    \justifies
    h(s_1) = h(s_2) \wedge R(t(s_1), t(s_2)) 
  \end{prooftree}\]
Two states $s_1,s_2\in K$ are {\em $c$-bisimilar}, notation
$s_1\bis^c s_2$, is defined by
\[s_1\bis^c s_2 :=\exists R \subseteq K\times K\, (R\mbox{ is a
  $c$-bisimulation and } R(s_1,s_2)).\]
\item A relation $Q\subseteq K\times K$ is a {\em $c$-apartness\/} if
  it satisfies the following rules
  \[\begin{prooftree}
    Q(t(s_1), t(s_2)) 
    \justifies
     Q(s_1, s_2)
  \end{prooftree}
  \qquad  \qquad
  \begin{prooftree}
    h(s_1) \neq h(s_2)
    \justifies
    Q(s_1, s_2) 
  \end{prooftree}
  \]
Two states $s_1,s_2\in K$ are {\em $c$-apart}, notation $s_1\apt^c
s_2$, is defined by
\[s_1\apt^c s_2 := \forall Q
\subseteq K\times K\, (\mbox{if }Q\mbox{ is a $c$-apartness, then } Q(s_1,s_2)).\]
\end{enumerate}
\end{defi}

%

Before we prove some generalities about bisimulation and apartness, we
now first treat the example of deterministic automata, DAs. A DA over
$A$ is given by a set of states, $K$, a transition function $\delta:
K\times A \rightarrow K$ and a function $f: K \rightarrow \{0,1\}$
denoting whether $q\in K$ is a final state or not. We write $2$ for
$\{0,1\}$ and we view, as usual in coalgebra, a DA as a coalgebra $c:
K \rightarrow K^A\times 2$, consisting of two maps $c =
\pair{\delta}{f}$ with $\delta \colon K \rightarrow K^A$ and $f \colon
K\rightarrow 2$. We use the standard notation for automata and write
$q \trans_{a} q'$ if $\delta(q)(a) = q'$ and $q\downarrow$ if $f(q) =
0$.

We now introduce the notions of bisimulation and apartness for
DAs. The first is well-known, the second less so. These notions can be
defined in a canonical way for a large set of functors on $\Set$. This
we will describe categorically in Section~\ref{ApartnesSec}. In the
Appendix, we will give an outline in logical-syntactic terms.

\begin{defi}\label{def.DFA}
Let $A$ be an alphabet and let $K$ be a set of states. A coalgebraic
map $\pair{c_1}{c_2} : K \rightarrow K^A\times 2$ gives rise to the
following notions of {\em bisimulation for $c$\/} and {\em apartness
  for $c$}.
\begin{enumerate}
\item A relation $R\subseteq K\times K$ is a {\em $c$-bisimulation\/}
  if it satisfies the following rule.
  \[\begin{prooftree}
    R(q_1, q_2)
    \justifies
    \forall a\in A\, \forall p_1,p_2 ( q_1 \trans_a p_1\wedge q_2 \trans_a p_2 \implies R(p_1,p_2)) \quad \wedge\quad
    q_1\downarrow  \;\Leftrightarrow\; q_2\downarrow
  \end{prooftree}\]
That two states $q_1,q_2\in K$ are {\em $c$-bisimilar}, notation
$q_1\bis^c q_2$, is defined by
\[q_1\bis^c q_2 \quad:=\quad \exists R \subseteq K\times K\, (R\mbox{
  is a $c$-bisimulation and } R(q_1,q_2)).\]
\item A relation $Q\subseteq K\times K$ is a {\em $c$-apartness\/} if
  it satisfies the following rules.
  \[\begin{prooftree}
    q_1 \trans_a p_1\qquad q_2 \trans_a p_2 \qquad Q(p_1,p_2) 
    \justifies
     Q(q_1, q_2)
  \end{prooftree}
  \qquad \qquad
  \begin{prooftree}
   q_1\downarrow \wedge \neg (q_2\downarrow)
    \justifies
    Q(q_1, q_2) 
  \end{prooftree}
  \qquad \qquad
  \begin{prooftree}
   \neg (q_1\downarrow) \wedge q_2\downarrow
    \justifies
    Q(q_1, q_2) 
  \end{prooftree}
  \]
As usual, rules are ``schematic'' in the free variables that occur in
it, so the left rule represents a separate rule for each $a\in A$.
That two states $q_1,q_2\in K$ are {\em $c$-apart}, notation
$q_1\apt^c q_2$, is defined by
\[q_1\apt^c q_2 \quad :=\quad\forall Q
\subseteq K\times K\, (\mbox{if }Q\mbox{ is a $c$-apartness, then } Q(q_1,q_2)).\]
\end{enumerate}
\end{defi}

In case the coalgebra $c$ is clear from the context, we will ignore
it.  In DAs, two states are bisimilar if and only if they are not
apart, which can easily be observed in the following example.

\begin{exa}
  Consider the DA given to the left below
$$\begin{array}{cc}
\hspace*{-1em}\vcenter
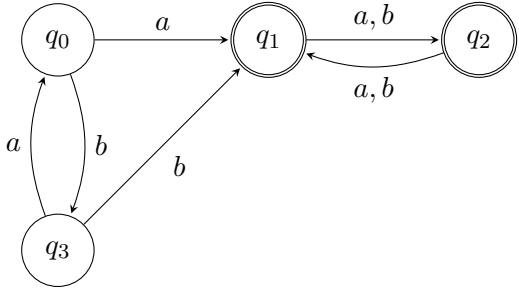
{\begin{tikzpicture}[>=stealth,shorten >=1pt,auto,node distance=2.8cm]
  \node[state] (q0)                 {$q_0$};
  \node[state, accepting] (q1) [right of=q0]   {$q_1$};
  \node[state, accepting] (q2) [right of=q1]   {$q_2$};
  \node[state] (q3) [below of=q0]   {$q_3$};

  \path[->]          (q0)  edge                   node {$a$} (q1);
  \path[->]          (q0)  edge   [bend left=20]  node {$b$} (q3);
  \path[->]          (q1)  edge                   node {$a,b$} (q2);
  \path[->]          (q2)  edge   [bend left=20]  node {$a,b$} (q1);
  \path[->]          (q3)  edge   [bend left=20]  node {$a$} (q0);
  \path[->]          (q3)  edge                   node [swap] {$b$} (q1);
  \end{tikzpicture}}\hspace*{-9em}
  &
\hspace*{-9em}\vcenter{\begin{prooftree}
    q_3 \transa q_0 \quad q_0 \transa q_1 \quad
    \begin{prooftree}
      \neg (q_0 \downarrow) \wedge\, q_1\downarrow
      \justifies
      Q(q_0,q_1)
    \end{prooftree}
    \justifies
     Q(q_3,q_0)
  \end{prooftree}}
\end{array}$$
  
\noindent A bisimulation is given by $q_1\sim q_2$. It can be shown that $q_0
\apt q_3$ because for every apartness $Q$ we have the derivation given
on the right.
\end{exa}

We see that ``being $c$-apart'', being the smallest relation
satisfying specific closure properties, is an inductive property. This
implies that the closure properties yield a {\em derivation system\/}
for proving that two elements are $c$-apart. This will be further
explored in the next section. In the example, we are basically using
this: we have proven $q_3 \apt q_0$ by giving a derivation.

A relation $Q$ is usually (e.g.\ see~\cite{TroelstraVanDalenII},
Chapter 8) called an apartness relation if it is irreflexive,
symmetric and co-transitive. As we have already used the terminology
``apartness relation'' for the dual of a bisimulation relation, we
shall, for the present paper, refer to these as ``proper apartness
relations''.

\begin{defi}\label{def:properapart}
  A relation $Q$ is called a {\em proper apartness relation\/} if it is
  \begin{itemize}
  \item {\em irreflexive}: $\forall x\, \neg Q(x,x)$,
  \item {\em symmetric}: $\forall x, y\, (Q(x,y)\implies Q(y,x))$,
    \item {\em co-transitive}: $\forall x, y,z\, (Q(x,y)\implies Q(x,z)\vee Q(z,y))$.
  \end{itemize}
\end{defi}

It is easy to see that inequality on a set is a proper apartness relation. The
following is a standard fact that relates equivalence relations and
proper apartness relations.

\begin{lem} \label{lem.apartnegequiv}
For $R$ a relation, $R$ is an equivalence relation if and only if
$\neg R$ is a proper apartness relation.
\end{lem}

\begin{proof}
The only interesting property to check is that $R$ is transitive iff
$\neg R$ is co-transitive. If $\neg R(x,y)$ and $R(x,z)$, then $\neg
R(z,y)$ by transitivity of $R$, so we have $\neg R(x,y) \implies \neg
R(x,z) \vee \neg R(z,y))$. The other way around, suppose $R(x,y)$ and
$R(y,z)$ and $\neg R(x,z)$. Then $\neg R(x,y) \vee \neg R(z,y)$ by
co-transitivity of $\neg R$, contradiction, so $R(x,z)$.
\end{proof}

Bisimulation and apartness for DAs and streams can be defined by
induction over the structure of the functor $F:\Set \rightarrow \Set$
that we consider the coalgebra for. In the case of DAs, we have $c: K
\rightarrow F(K)$ with $F(X) = X ^A\times 2$ and for streams, we have
$c: K \rightarrow F(K)$ with $F(X) = A\times X$. The general
definition in category-theoretic terms can be found in Section
\ref{CoalgLiftSec}. A purely logical-syntactic presentation can be
found in the Appendix. 

\begin{lem}\label{lem.bisimandaprt}
  We have the following result relating bisimulation and apartness for
  the case of DAs and streams (but it also applies to the general case
  treated in the Appendix).
  \begin{enumerate}
  \item $R$ is a bisimulation if and only if $\neg R$ is an apartness.
  \item The relation $\bis$ is the union of all bisimulations, $\bis
    \;=\, \bigcup\{R \mid R \mbox{ is a bisimulation}\}$, and it is
    itself a bisimulation.
  \item The relation $\apt$ satisfies $\apt \;=\, \bigcap\{Q \mid Q
    \mbox{ is an apartness relation}\}$, and is thus the intersection
    of all apartness relations; it is itself also an apartness
    relation.
  \item $\mathord{\bis} = \neg \mathord{\apt}$.
  \end{enumerate}
\end{lem}

\begin{proof}
  We show the first in some detail for the case of DAs (Definition
  \ref{def.DFA}). It rests on some simple logical equivalences. That $R$ is a $c$-bisimulation is equivalent to: 
  \begin{eqnarray*}
& \Leftrightarrow&\forall q_1,q_2 ( R(q_1, q_2)
    \implies \forall a\in A\, R(c_1(q_1)(a), c_1(q_2)(a)) \wedge c_2(q_1) = c_2(q_2))\\
& \Leftrightarrow&\forall q_1,q_2 ( \neg \forall a\in A\, R(c_1(q_1)(a), c_1(q_2)(a)) \wedge c_2(q_1) = c_2(q_2)) 
    \implies \neg R(q_1, q_2) \\
& \Leftrightarrow& \forall q_1,q_2 ( \exists a\in A\, \neg R(c_1(q_1)(a), c_1(q_2)(a)) \vee c_2(q_1) \neq c_2(q_2)) \implies  \neg R(q_1, q_2)\\
& \Leftrightarrow& \forall q_1,q_2 ( \exists a\in A\, \neg R(c_1(q_1)(a), c_1(q_2)(a))\! \implies\!  \neg R(q_1, q_2) \wedge (c_2(q_1) \neq c_2(q_2) \! \implies \! \neg R(q_1, q_2))),    
  \end{eqnarray*}
  which states that $\neg R$ is a $c$-apartness.

The other items are easily verified: if $R_1$ and $R_2$ are
bisimulations, then $R_1\cup R_2$ is also a bisimulation, and if $Q_1$
and $Q_2$ are apartness relations, then $Q_1\cap Q_2$ is also an
apartness relation.
\end{proof}

\begin{rem}\label{rem:derivrules}
A relation $R$ is a $c$-bisimulation in case it satisfies a specific {\em closure
  property\/} that is given in Definitions \ref{def.streams},
\ref{def.DFA} via a rule that $R$ should satisfy. Similarly, there is
a closure property that defines when $Q$ is a $c$-apartness (also
given via a rule that $Q$ should satisfy).

To prove that $s$ and $t$ are $c$-bisimilar, we need to find an $R$
that satisfies the rules for $c$-bisimulations such that $R(s,t)$
holds. Dually, to prove that $s$ and $t$ are $c$-apart, we need to
show that $Q(s,t)$ holds for every $Q$ that satisfies the rules for
$c$-apartness. This means that we can use the rules for
being a $c$-apartness as the {\em derivation rules\/} of the proof system for
proving $s\apt^c t$: we have $s\apt^c t$ if and only if there is a finite derivation of $s\apt^c t$ using these rules.

So, for apartness, the rules that define ``$Q$ is a $c$-apartness''
can be used as the derivation rules for proving $s\apt^c t$. This is
obviously not the case for bisimilarity. There the rules just
represent the closure properties that $R$ should satisfy to be a
$c$-bisimulation\footnote{One could think of using the rules for
  bisimulation as ``infinitary proof rules'', where one allows
  infinite derivations of some form, but we will not expand on that
  here.}
\end{rem}

In Sections~\ref{CoalgLiftSec} and~\ref{ApartnesSec} we will give a
more general categorical picture of bisimulation and apartness on
coalgebras.

\subsection{Apartness in constructive mathematics}
The notion of apartness is standard in constructive real analysis and
goes back to Brouwer, with Heyting giving the first axiomatic
treatment in \cite{Heyting}. (See also e.g.~\cite{TroelstraVanDalenII} Chapter 8.)
The observation is that, if one reasons in constructive logic, the
primitive notion for real numbers is apartness: if two real numbers are
apart, this can be positively decided in a finite number of steps,
just by computing better and better approximations until one
positively knows an $\epsilon$-distance between them. Then equality on
real numbers is defined as the negation of apartness: $x =y := \neg(x
\# y)$.

As a matter of fact, one can start from apartness and define equality
using its negation, and then build up the real numbers axiomatically
from there. This is done in \cite{GeuversNiqui}, where an axiomatic
description of real numbers is given and it is shown how Cauchy
sequences over the rationals form a model of that axiomatization, all
in a constructive setting, i.e.\ without using the excluded middle
rule. If one assumes apartness $\#$ to be a proper apartness (as in
our Definition \ref{def:properapart}), the defined equality is an
equivalence relation.

In the setting of the present paper, these constructive issues do not
play a role, because we reason classically.  There is one point to
make, which is the issue of congruence, which has been studied in
depth in the context of process theory
\cite{BaetenBastenReniers,Fokkink}. Then the question is if, in a
theory of terms describing processes, with a notion of bisimilarity
describing a semantic equivalence of the terms as labelled transition
systems, bisimulation is preserved by the operators of the
theory. Simply put: if $q_1 \bis p_1$ and $q_2\bis p_2$, is it the
case that $f(q_1,q_2)\bis f(p_1,p_2)$? In constructive analysis, if
one starts from apartness and defines equality as its negation, the
corresponding notion is {\em strong extensionality}.

\begin{defi}\label{def.strext}
  A function $f: K\times K \rightarrow K$ is {\em strongly
    extensional\/} if
  \[\forall x_1,x_2,y_1,y_2 \in K( f(x_1,x_2) \# f(y_1,y_2)
  \implies x_1 \# y_1 \vee x_2 \# y_2).\]
  A relation $R \subseteq K\times K$ is {\em strongly
    extensional\/} if
  \[\forall x_1,x_2,y_1,y_2 \in K( R(x_1,x_2)
  \implies R(y_1,y_2) \vee x_1 \# y_1 \vee x_2 \# y_2).\]
\end{defi}

It is easily checked that, if one defines an equivalence relation
$\sim$ as the negation of $\#$, then strong extensionality implies
congruence with respect to $\sim$. So, if we wish to deal with process
theories in terms of apartness, we will have to require operations and
relations to be strongly extensional. It turns out that weaker forms
of bisimulation (e.g.\ branching bisimulation) are not congruences,
and therefore one considers {\em rooted branching bisimulation}. In
Section \ref{sec:usingapartness} we will briefly study its complement,
{\em rooted branching apartness\/} and the connection between congruence and strong
extensionality.

\section{Weak and branching bisimulation}
\label{sec:weakbran}
We now apply the techniques that we have seen before to weak and
branching bisimulation. We do not give a categorical treatment,
because the functors proposed for weak~\cite{SokolovaVinkWoracek} and
branching~\cite{BeoharKupper} bisimulation are not so easy to work
with. Instead, we use the definition of ``bisimulation'' (for a
specific type of system) to directly define the notion of
``apartness'' as its negation, and thereby we define a derivation
system for apartness. Then, two states $s$ and $t$ are (weakly,
branching) apart iff they are not (weakly, branching) bisimilar. We
also apply our definitions in a simple example to show how apartness
(and thereby the absence of a bisimulation) can be proved.

We also rephrase some known results about branching bisimulation in
terms of apartness, notably we reprove the stuttering property for
branching bisimulation and the fact that branching bisimulation is an
equivalence relation by rephrasing these results in terms of branching
apartness. In the known proofs of these results, the notion of {\em
  semi-branching bisimulation\/} is used. Here we use a notion of {\em
  semi-branching apartness\/} for similar purposes.  Finally we look
into applications of the derivation system for actually deriving that
two states in an LTS are branching apart (and therefore not branching
bisimilar) and we suggest some new rules, using both apartness and
bisimulation, that may be useful for analyzing algorithms for
branching bisimulation.

The systems we focus on are {\em labelled transition systems},
LTSs. An LTS is a tuple $(X,\Atau, \trans)$, where $X$ is a set of
states, $\Atau = A\cup\{\tau\}$ is a set of actions (containing the
special ``silent action'' $\tau$), and $\trans \;\subseteq X \times
\Atau\times X$ is the {\em transition relation}. We write $q_1 \transu
q_2$ for $(q_1,u,q_2) \in\; \trans$ and we write $\transtt$ to denote
the reflexive transitive closure of $\transt$. So $q_1 \transtt q_2$
if $q_1 \transt \ldots \transt q_2$ in zero or more $\tau$-steps.

\begin{convention}
We will reserve $q_1 \transa q_2$ to denote a transition with an
$a$-step with $a\in A$ (so $a\neq \tau$).
\end{convention}

First we recapitulate the standard definitions of labelled transition
system and weak and branching bisimulation. We do this in a
``rule'' style. The standard definition of $R\subseteq
X\times X$ being a \emph{weak} bisimulation relation is that
we have, for all
$q,p,q' \in X$ and all $a \in A$,
\begin{eqnarray*}
 R(q,p) \wedge q \transt q'&\implies& \exists p'(p\transtt p' \wedge
 R(q', p'))\\ 
 R(q,p) \wedge q \transa q'&\implies& \exists
 p',p'',p'''(p\transtt p' \transa p'' \transtt p'''\wedge R(q',p''')),
\end{eqnarray*}
and also the symmetric variants of these two properties:
\begin{eqnarray*}
 R(p,q) \wedge q \transt q'&\implies& \exists p'(p\transtt p' \wedge
 R(p', q'))\\ 
 R(p,q) \wedge q \transa q'&\implies& \exists
 p',p'',p'''(p\transtt p' \transa p'' \transtt p'''\wedge R(p''',q')),
\end{eqnarray*}
Many rules in the rest of this paper have symmetric variants, like
branching bisimulation above. We will not give these explicitly, but
just refer to them as the ``symmetric variants'' of the rules.

We will rephrase the properties of weak/branching bisimulation
(equivalently) as rules. These look uncommon for
bisimulation, but will turn out to be useful when we look at their
inverse, apartness.

\begin{defi}\label{def.weakbran}
A relation $R\subseteq X\times X$ on a LTS $(X,\Atau, \trans)$ is a
{\em weak bisimulation relation\/} if it the
following two rules and their symmetric variants hold for $R$.
\[
\begin{prooftree}
  q \transt q' \qquad R(q,p)
  \justifies
  \exists p'(p\transtt p' \wedge R(q', p'))
  \using{\weakbist}
\end{prooftree}
\]

\[\begin{prooftree}
q \transa q'\qquad R(q,p)
\justifies
\exists p',p'',p'''(p\transtt p' \transa p'' \transtt p'''\wedge R(q',p'''))
\using{\weakbisa}
\end{prooftree}
\]
The states $q, p$ are {\em weakly bismilar}, notation $q\weakbis p$ if
and only if there exists a weak bisimulation relation $R$ such that
$R(q,p)$.

A relation $R\subseteq X\times X$ is a {\em branching bisimulation
  relation\/} if the
following two rules and their symmetric variants hold for $R$.
\[
\begin{prooftree}
  q \transt q'\qquad R(q,p)
  \justifies
  R(q', p) \vee \exists p',p''(p\transtt p' \transt p'' \wedge R(q,p')\wedge
  R(q',p''))
  \using{\branbist}
\end{prooftree}
\]

\[\begin{prooftree}
  q \transa q' \qquad R(q,p)
  \justifies
  \exists p',p''(p\transtt p' \transa p'' \wedge R(q,p')\wedge
  R(q',p''))
    \using{\branbisa}
\end{prooftree}
\]
The states $q, p$ are {\em branching bisimilar}, notation $q\branbis p$
if and only if there exists a branching bisimulation relation $R$ such
that $R(q,p)$.
\end{defi}

It is well-known that weak bisimulation is really weaker than
branching bisimulation (if $s\branbis t$, then $s\weakbis t$, but in
general not the other way around) and that various efficient
algorithms for checking branching bisimulation exist
(\cite{GrooteVaandrager,Jansenetal}). Here we wish to analyze these
notions by looking at their opposite: {\em weak apartness\/} and {\em
  branching apartness}.

\begin{defi}\label{def.weakapart}
  Given a labelled transition system $(X,\Atau, \trans)$, we say that
  $Q\subseteq X\times X$ is a {\em weak apartness relation\/} in case
  the following rules
  hold for $Q$.
\[
  \begin{prooftree}
    Q(p, q)
    \justifies
    Q(q, p)
    \using{\symm}
  \end{prooftree}
\]

\[
  \begin{prooftree}
    q \transt q' \qquad \forall p'(p\transtt p' \implies Q(q',p'))
    \justifies
    Q(q,p)
    \using{\weakapt}
  \end{prooftree}
\]

\[
  \begin{prooftree}
    q \transa q' \qquad \forall p',p'',p'''(p\transtt p' \transa p'' \transtt p'''\implies Q(q',p'''))
    \justifies
    Q(q,p)
        \using{\weakapa}
  \end{prooftree}
\] 

The states $q$ and $p$ are {\em weakly apart}, notation $q\weakap p$,
if for all weak apartness relations $Q$, we have $Q(q,p)$.

\end{defi}

The relation of ``being weakly apart'' is itself a weak apartness
relation: it is the smallest weak apartness relation, so we have an
inductive definition of ``being weakly apart'', using a derivation
system. We express this explicitly in the following Corollary to the
Definition.

\begin{cor}\label{cor.weakapart}
  Given a labelled transition system $(X,\Atau, \trans)$, and $q,p \in
  X$, we have $q\weakap p$ if and only if this can be derived using
  the following derivation rules.
\[
  \begin{prooftree}
    p \weakap q
    \justifies
    q\weakap p
    \using{\symm}
  \end{prooftree}
\]
  
  \[
  \begin{prooftree}
    q \transt q' \qquad \forall p'(p\transtt p' \implies q' \weakap p')
    \justifies
    q\weakap p
    \using{\weakapt}
  \end{prooftree}
\]

\[
  \begin{prooftree}
    q \transa q' \qquad \forall p',p'',p'''(p\transtt p' \transa p'' \transtt p'''\implies q' \weakap p''')
    \justifies
    q\weakap p
        \using{\weakapa}
  \end{prooftree}
\] 
\end{cor}

\begin{rem}[Also see Remark \ref{rem:derivrules}] \label{rem:weakapt_derivrules}
The notions
of weak bisimulation and weak apartness are defined using closure
properties that a relation should satisfy. As weak apartness is an
inductive notion, the rules that define the closure property for weak
apartness can be used as the derivation rules of a proof system
to derive $q\weakap p$. More precisely: we have $q\weakap p$ if and
only if this can be derived using a finite derivation with the rules
of Corollary \ref{cor.weakapart}.  Again, this is not the case for
weak bisimilarity.
\end{rem}

We now define the notion of {\em branching apartness}.

\begin{defi}\label{def.branapart}  Given a labelled transition system $(X,\Atau, \trans)$, we say that
  $Q\subseteq X\times X$ is a {\em branching apartness\/} in case the following rules hold for $Q$.

\[
  \begin{prooftree}
    Q(p, q)
    \justifies
    Q(q, p)
    \using{\symm}
  \end{prooftree}
\]

\[
  \begin{prooftree}
    q \transt q' \qquad  Q(q',p) \qquad \forall p',p''(p\transtt p' \transt p''\implies Q(q,p')\vee Q(q',p''))
    \justifies
    Q(q,p)
    \using{\branapt}
  \end{prooftree}
\]

\[
  \begin{prooftree}
    q \transa q' \qquad \forall p',p''(p\transtt p' \transa p''\implies Q(q,p')\vee Q(q',p''))
    \justifies
    Q(q,p)
        \using{\branapa}
  \end{prooftree}
\] 

The states $q$ and $p$ are {\em branching apart}, notation $q\branap p$,
if for all branching apartness relations $Q$, we have $Q(q,p)$.

\end{defi}

Again, being branching apart is an inductive definition (it is the
smallest branching apartness relation), so we have a derivation
system. We express this explicitly in the following Corollary to the
Definition, where again Remark \ref{rem:weakapt_derivrules} applies.

\begin{cor}\label{cor.branap}
  Given a labelled transition system $(X,\Atau, \trans)$, and $q,p \in
  X$, we have $q\branap p$ if and only if this can be derived using
  the following derivation rules.
  \[
  \begin{prooftree}
    p \branap q
    \justifies
    q\branap p
    \using{\symm}
  \end{prooftree}
  \]
  
  \[
  \begin{prooftree}
    q \transt q' \qquad  q' \branap p \qquad \forall p',p''(p\transtt p' \transt p''\implies q \branap p'\vee q' \branap p'')
    \justifies
    q \branap p
    \using{\branapt}
  \end{prooftree}
\]

\[
  \begin{prooftree}
    q \transa q' \qquad \forall p',p''(p\transtt p' \transa p''\implies q \branap p'\vee q' \branap p'')
    \justifies
    q \branap p
        \using{\branapa}
  \end{prooftree}
\] 
\end{cor}

\begin{rem}[A note on symmetry]\label{rem:symm}
  In the rules, e.g.\ of Definition \ref{def.branapart} and Corollary
  \ref{cor.branap}, there is a choice of adding symmetry as a rule, or
  adding symmetric variants of the rules. In our presentation, we
  choose to add symmetry as a rule. In the literature on bisimulation,
  it is standard to add symmetric variants of the rules, and then it
  can be shown that the relations themselves are symmetric. To be
  clear, the symmetric variants of the rules
  of Corollary \ref{cor.branap} would be as follows.  
  \[
  \begin{prooftree}
    p \transt p' \qquad  q \branap p' \qquad \forall q',q''(q\transtt q' \transt q''\implies q' \branap p\vee q'' \branap p')
    \justifies
    q \branap p
    \using{\branapt'}
  \end{prooftree}
\]

\[
  \begin{prooftree}
    p \transa p' \qquad \forall q',q''(q\transtt q' \transa q''\implies q' \branap p\vee q'' \branap p')
    \justifies
    q \branap p
        \using{\branapa'}
  \end{prooftree}
\] 
and then one can prove that (without rule ($\symm$)), the relation  $\branap$ is symmetric.

In the following, we will regularly prove properties about an
apartness relation by {\em induction on the derivation\/} and then of
course it matters which rules one has chosen. We found that having
symmetry as a rule, and not a slightly informal ``symmetric
duplication'' of all rules is a bit more clear and concise.  In fact,
for the proofs that are given below, it doesn't really matter what
rules one has chosen: symmetry as a rule, or symmetry ``built in'' by
adding the symmetric variants of the rules. The induction proofs that
follow are mostly symmetric in either side of the apartness sign, with
one notable exception, and that is the stuttering property, Lemma
\ref{lem.apstut}.
\end{rem}

We now show how to use apartness on a few simple well-known
examples. We show how we can {\em derive\/} that two states are
branching apart (i.e.\ not branching bisimilar) by giving a derivation
of this fact using the rules for $\branap$.

\begin{exa}
  We describe two LTSs from~\cite{DeNicolaVaandrager} that serve as
  examples to show the difference between weak and branching
  bisimulation. We apply our apartness definitions to show the
  difference between $\weakap$ and $\branap$.  The LTS on the left
  consists of states $\{s,s_1,s_2,s_3,s_4, r, r_1,r_2, r_3\}$ and the
  point is that $s\branap r$, while $s \weakbis r$. The LTS on the
  right consists of states $\{q, q_1,q_2, q_3,
  q_4,q_5,p,p_1,p_2,p_3,p_4\}$ and the point is that $q\branap p$,
  while $q \weakbis p$.

  \medskip
  
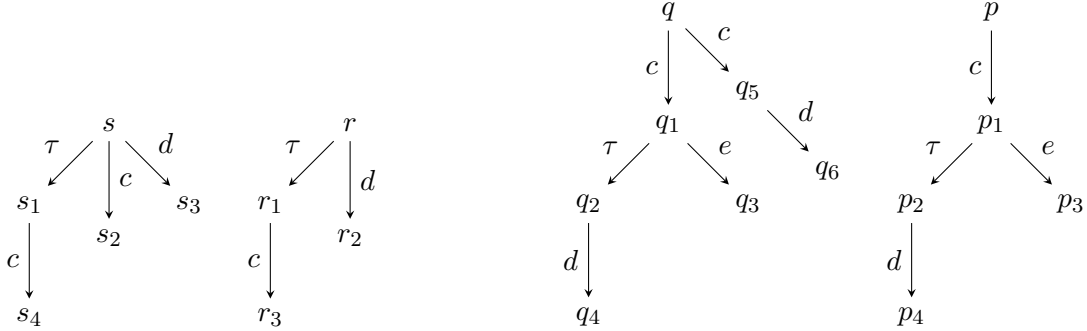
    \begin{tabular}[t]{ccccc}
\begin{tikzpicture}[>=stealth,node distance=1.5cm,auto]
    \node    (s)                       {$s$};
    \node    (s1) [below left of = s]  {$s_1$};
    \node    (s4) [below of = s1]      {$s_4$};
    \node    (s3) [below right of = s] {$s_3$};
    \node    (s2) [below of = s]       {$s_2$};

    \path[->]
        (s)  edge                    node[swap] {$\tau$}  (s1)
             edge                    node {$d$}     (s3)
             edge                    node {$c$}     (s2)
        (s1) edge                    node[swap] {$c$}     (s4);
\end{tikzpicture}
&
\begin{tikzpicture}[>=stealth,node distance=1.5cm,auto]
    \node   (r)                       {$r$};
    \node   (r1) [below left of = r]  {$r_1$};
    \node   (r3) [below of = r1]      {$r_3$};
    \node   (r2) [below of = r] {$r_2$};

    \path[->]
        (r)  edge                    node[swap] {$\tau$}  (r1)
             edge                    node {$d$}     (r2)
        (r1) edge                    node[swap] {$c$}     (r3);
\end{tikzpicture}
&\qquad
\qquad
\qquad&
\begin{tikzpicture}[>=stealth,node distance=1.5cm,auto]
    \node  (q)                        {$q$};
    \node  (q1) [below of = q]        {$q_1$};
    \node  (q2) [below left of = q1]  {$q_2$};
    \node  (q5) [below right of = q]  {$q_5$};
    \node  (q6) [below right of = q5] {$q_6$};
    \node  (q3) [below right of = q1] {$q_3$};
    \node  (q4) [below of = q2]       {$q_4$};

    \path[->]
    (q)  edge                   node[swap] {$c$}       (q1)
         edge                   node {$c$}       (q5)
    (q1) edge                   node[swap] {$\tau$}  (q2)
         edge                   node {$e$}       (q3)
    (q2) edge                   node[swap] {$d$}       (q4)
    (q5) edge                   node {$d$}       (q6);
\end{tikzpicture}
&
\begin{tikzpicture}[>=stealth,node distance=1.5cm,auto]
    \node  (p)                        {$p$};
    \node  (p1) [below of = p]        {$p_1$};
    \node  (p2) [below left of = p1]   {$p_2$};
    \node  (p3) [below right of = p1] {$p_3$};
    \node  (p4) [below of = p2]       {$p_4$};

    \path[->]
    (p)  edge                    node[swap] {$c$}       (p1)
    (p1)  edge                   node[swap] {$\tau$}  (p2)
          edge                   node {$e$}       (p3)
    (p2) edge                    node[swap] {$d$}       (p4);
\end{tikzpicture}

    \end{tabular}
In the LTS on the left, we have $s\branap r_1$, because $s$ can do a
$d$-step, while $r_1$ can not.  Therefore, $s\branap r$, because
$s\trans_c s_2$ and the only possible $c$-step from $r$ is $r\transtt
r_1\trans_c r_3$, and $s\branap r_1$. Given that we now have a derivation system, we can also give a derivation of $s\branap r$:
\[\begin{prooftree}
  s \trans_c s_2 \qquad r \transt r_1 \trans_c r_3 \qquad\qquad
  \begin{prooftree}
    \begin{prooftree}
      s \trans_d s_3
      \justifies
      s \branap r_1
    \end{prooftree}
  \justifies
        s \branap r_1 \vee s_2 \branap r_3
  \end{prooftree}
  \justifies
  s\branap r
\end{prooftree}\]
On the other hand we have $s\weakbis r$. This can be seen by the
weak bisimulation $\sim$ given by the following equivalence classes: $\{s, r\}$, $\{s_1,r_1\}$, $\{s_2,s_4,r_3\}$, $\{s_3,r_2\}$. This is indeed a weak bisimulation following
Definition~\ref{def.weakbran}. A different way to prove $s\weakbis r$
is by showing $\neg s\weakap r$, which can be achieved by proving that
there is no derivation of $s\weakap r$. This is more involved, as we
have to reason about all possible derivations of $s\weakap r$. The
only relevant candidate is below, which fails on finding a derivation
of $s_2\weakap s_3$ (which does not exist).
\[\begin{prooftree}
  s \trans_c s_2 \qquad r \transt r_1 \trans_c r_3 \qquad\qquad
  \begin{prooftree}
  ??
  \justifies
  s_2 \weakap r_3
  \end{prooftree}
  \justifies
  s\weakap r
\end{prooftree}\]

In the LTS on the right, we have $q_5 \branap p_1$, because $q_5$
cannot do an $e$-step. Therefore, $q \branap p$, because $q\trans_c
q_5$ and the only $c$-step from $p$ leads to $p_1$ and $q_5 \branap
p_1$. Also here, we can give a derivation:
\[\begin{prooftree}
  q \trans_c q_5 \qquad p \trans_c p_1 \qquad\qquad
  \begin{prooftree}
    \begin{prooftree}
       p_1 \trans_e p_3
      \justifies
      q_5 \branap p_1
    \end{prooftree}
  \justifies
        q \branap p \vee q_5 \branap p_1
  \end{prooftree}
  \justifies
  q\branap p
\end{prooftree}\]
\end{exa}

The notions of weak, resp.\ branching, apartness and weak,
resp.\ branching, bisimulation relate in the standard way we have seen
before in Section \ref{sec:bisapt}: $R$ is a weak (branching)
apartness if and only if $\neg R$ is a weak (branching) bisimulation.
This also implies that we can transfer properties from
(weak/branching) bisimulation to (weak/branching) apartness and vice
versa. In the next Section, we show how we can use apartness to proved
results about bisimulation. We now summarize the results that relate
bisimulation and apartness in a couple of Lemmas.

\begin{lem}\label{lem.apartinvequiv}
A relation $R$ over an LTS is a weak (resp.\ branching) bisimulation
if and only if $\neg R$ is a weak (resp.\ branching) apartness.
\end{lem}

\begin{proof}
The proofs are by some standard logical manipulations, similar to the
proof of Lemma \ref{lem.bisimandaprt}. To simplify the work, it is
easiest to first replace the rule ($\symm$) by the ``symmetric
variants'' of the other rules, as discussed in Remark \ref{rem:symm}.
\end{proof}

We have
$\weakbis \;=\, \bigcup\{R \mid R\mbox{ is a weak bisimulation}\}$ and
similarly for $\branbis$ and it is straightforward to verify that
$\weakbis$ is itself a weak bisimulation (and similarly for
$\branbis$). For apartness we have the same result: $\weakap \;=\,
\bigcap\{Q \mid Q\mbox{ is a weak apartness}\}$, and similarly for
$\branap$. The last part of the Lemma follows from
\begin{eqnarray*}
\neg(q\weakbis p) &\Leftrightarrow& \neg\exists R (R \mbox{ is a
    weak bisimulation} \wedge R(q,p))\\
    &\Leftrightarrow& \forall R (R \mbox{ is a weak bisimulation} \implies
\neg R(q,p))\\
    &\Leftrightarrow& \forall Q (Q \mbox{ is a weak apartness} \implies Q(q,p))\\
&\Leftrightarrow& q\weakap p.
\end{eqnarray*}
This results in the following Lemma.

\begin{lem}
  \label{lem.largestsmallest}
  \begin{enumerate}
  \item $\weakbis$ (resp.\ $\branbis$) is the largest weak (resp.\ branching)
    bisimulation.
  \item $\weakap$ (resp.\ $\branap$) is the smallest weak (resp.\ branching)
    apartness.
  \item $\weakap = \neg\weakbis$ and $\branap = \neg\branbis$.
  \end{enumerate}
\end{lem}


\subsection{Using apartness to prove results about bisimulation}
The first result we prove is that weak apartness is included in
\label{sec:apartbisim}
branching apartness, which implies the well-known result that
branching bisimulation is included in weak bisimulation. The
interesting aspect is that we prove these results by induction (on the
derivation). Then we will prove co-transitivity of branching apartness
(which implies transitivity of branching bisimulation). We introduce
semi-branching apartness as a means to prove a stuttering property and
some other basic properties (for semi-branching apartness), from which
we can conclude that semi-branching and branching apartness are the
same, from which we derive co-transitivity.

\begin{lem}\label{lem.weakapisbisap}
If $s \weakap t$, then $s\branap t$.
\end{lem}

\begin{proof}
By induction on the derivation of $s \weakap t$, where we distinguish
cases according to the last rule.
\begin{itemize}
  \item Last rule is ($\weakapt$)
  \[
  \begin{prooftree}
    q \transt q' \qquad \forall p'(p\transtt p' \implies q' \weakap p')
    \justifies
    q\weakap p
    \using{\weakapt}
  \end{prooftree}
\]
By induction we have $\forall p'(p\transtt p' \implies q' \branap
p')$, which implies $q'\branap p$ and $\forall p', p''(p\transtt
p'\transt p'' \implies q \branap p'\vee q' \branap p'')$, which are
the hypotheses for the rule ($\branapt$), so we conclude $q\branap p$
by the rule ($\branapt$).

\item  Last rule is ($\weakapa$)
\[
  \begin{prooftree}
    q \transa q' \qquad \forall p',p'',p'''(p\transtt p' \transa p'' \transtt p'''\implies q' \weakap p''')
    \justifies
    q\weakap p
        \using{\weakapa}
  \end{prooftree}
\] 
By induction we have $\forall p',p'',p'''(p\transtt p' \transa p''
\transtt p'''\implies q' \branap p''')$, which implies $\forall p',p''(p\transtt p' \transa p''\implies q \branap p'\vee q' \branap p'')$, which is the hypothesis for the rule ($\branapa$),
so we conclude $q\branap p$ by the rule ($\branapa$).
\end{itemize}
\end{proof}

It is well-known from the literature that the relations $\weakbis$ and
$\branbis$ are equivalence relations. For $\weakbis$, the proof is
in~\cite{Milner}.  For $\branbis$, the proof is remarkably subtle, as
it is not the case in general that, if $R_1$ and $R_2$ are branching
bisimulations, then $R_1 \after R_2$ is a branching bisimulation. In
\cite{Basten} the transitivity of $\branbis$ is proven (and thereby
that $\branbis$ is an equivalence relation), using the notion of {\em
  semi-branching bisimulation}.  In~\cite{GlabbeekWeijland,Basten},
semi-branching bisimulation is also used to prove the so called {\em
  stuttering property}.  The results from those papers can also be
cast in terms of apartness, which we will do now.  We prove that
$\branbis$ is an equivalence relation by proving that $\branap$ is a
proper apartness relation and using the fact that $\branbis$ is the
complement of $\branap$. Similarly we prove an {\em apartness
  stuttering property\/} for $\branap$ and conclude the stuttering
property for $\branbis$ from that. It turns out that, for proving
co-transitivity of $\branap$ (and also stuttering) we need a notion of
{\em semi-branching apartness}, which is comparable to the complement
of the notion of semi-branching bisimulation of
~\cite{GlabbeekWeijland,Basten} (but slightly different). We introduce
those notions first.

\begin{defi}\label{def:semibapt}
A relation $Q\subseteq X\times X$ is a {\em semi-branching
  apartness\/} in case the following derivation rules hold for $Q$. (So $\sbranapt$ replaces the rule
$\branapt$.)
\[
  \begin{prooftree}
    Q(p, q)
    \justifies
    Q(q, p)
    \using{\symm}
  \end{prooftree}
\]

\[
  \begin{prooftree}
    q \transt q' \qquad  Q(q',p)
    \qquad \forall p',p''(p\transtt p' \transt p''\implies Q(q',p'') \vee (Q(q,p')\wedge Q(q,p'') ))
    \justifies
    Q(q,p)
    \using{\sbranapt}
  \end{prooftree}
\]

\[
  \begin{prooftree}
    q \transa q' \qquad \forall p',p''(p\transtt p' \transa p''\implies Q(q,p')\vee Q(q',p''))
    \justifies
    Q(q,p)
        \using{\branapa}
  \end{prooftree}
\] 
\noindent The states $q$ and $p$ are {\em semi-branching apart}, notation $q\sbranap p$,
if for all semi-branching apartness relations $Q$, we have $Q(q,p)$.
\end{defi}

So the rules $\symm$ and $\branapa$ are the same as for branching
bisimulation of Definition \ref{def.branapart}, and only the rule for
$\tau$-steps has been modified. Note that in particular, to derive
$Q(q,p)$ from $q\transt q'$, we need to prove $Q(q',p)$ first.

\begin{cor}\label{cor.semibapt}
The states $q$ and $p$ are semi-branching apart, $q\sbranap p$, if
this can be derived from the following rules.
  \[
  \begin{prooftree}
    p \sbranap q
    \justifies
    q\sbranap p
    \using{\symm}
  \end{prooftree}
  \]

  \[
  \begin{prooftree}
    q \transt q' \qquad  q'\sbranap p
    \qquad \forall p',p''(p\transtt p' \transt p''\implies q'\sbranap p'' \vee (q\sbranap p'\wedge q \sbranap p'' ))
    \justifies
    q \sbranap p
    \using{\sbranapt}
  \end{prooftree}
\]

\[
  \begin{prooftree}
    q \transa q' \qquad \forall p',p''(p\transtt p' \transa p''\implies q \sbranap p'\vee q' \sbranap p'')
    \justifies
    q \sbranap p
        \using{\branapa}
  \end{prooftree}
\] 

\end{cor}

We also define the dual (complement) notion of a semi-branching
bisimulation relation.

\begin{defi}\label{def:semibis}
A relation $R\subseteq X\times X$ is a {\em semi-branching bisimulation
  relation\/} if the
following two derivation rules and the symmetry rule hold for $R$.
\[
\begin{prooftree}
  q \transt q'\qquad R(q,p)
  \justifies
  R(q',p) \vee \exists p',p''(p\transtt p' \transt p'' \wedge R(q',p'') \wedge (R(q,p')\vee R(q,p'')))
  \using{\sbranbist}
\end{prooftree}
\]

\[\begin{prooftree}
  q \transa q' \qquad R(q,p)
  \justifies
  \exists p',p''(p\transtt p' \transa p'' \wedge R(q,p')\wedge
  R(q',p''))
    \using{\sbranbisa}
\end{prooftree}
\]
The states $q, p$ are {\em semi-branching bisimilar}, notation $q\sbranbis p$
if and only if there exists a semi-branching bisimulation relation $R$ such
that $R(q,p)$.
\end{defi}


It can again be shown that $Q$ is a semi-branching apartness if and
only if $\neg Q$ is a semi-branching bisimulation. Using this and the
fact that $\sbranap$ is the smallest semi-branching apartness and
$\sbranbis$ is the largest semi-branching bisimulation, we obtain that
$q\sbranap p \Leftrightarrow \neg(q\sbranbis p)$.

Our definition of semi-branching bisimulation is slightly different
from the one in \cite{Basten} and \cite{GlabbeekWeijland}, but it can
be shown that they are equivalent.

The rest of this section will be devoted to proving the co-transitivity of $\branap$ (and thereby that $\branbis$ is an equivalence relation) in the following steps.
\begin{enumerate}
\item We prove that $q \sbranap p \implies q\branap p$: Lemma
  \ref{lem.sbranbran}.
\item We prove a number of basic Lemmas for $\sbranap$; typically
  useful results we would also like to have for $\branap$, but we
  can't obtain directly for $\branap$: Lemma \ref{lem.branapI} and Corollary \ref{cor.branapII}
\item We prove the {\em apartness stuttering property\/} for
  $\sbranap$: Lemma \ref{lem.apstut}.
\item We prove that $q \branap p \implies q\sbranap p$, using the
  apartness stuttering property, and we conclude that $\branap \quad=\quad
  \sbranap$: Lemma \ref{lem.bransbran}.
\item We prove co-transitivity for $\branap$, using the basic lemmas mentioned above.
\end{enumerate}
Many of the proofs will proceed by induction on the derivation, where
we use the apartness as an inductively defined relation (defined via
derivation rules). For one of the basic Lemmas under (2) we will move
over to the ``bisimulation view'', as the result seems easier to obtain
there.

\begin{lem}\label{lem.sbranbran}
For all states $q,p$, $q \sbranap p \implies q\branap p$.
\end{lem}

\begin{proof}
  By induction on the derivation of $q \sbranap p$. The only interesting case is when the last rule applied is ($\sbranapt$).
\[
  \begin{prooftree}
    q \transt q' \qquad  q'\sbranap p
    \qquad \forall p',p''(p\transtt p' \transt p''\implies q'\sbranap p'' \vee (q\sbranap p'\wedge q \sbranap p'' ))
    \justifies
    q \sbranap p
    \using{\sbranapt}
  \end{prooftree}
\]
We have $q \transt q'$ and by induction hypothesis $q' \branap p$ and
$\forall p',p''(p\transtt p' \transt p''\implies q'\branap p'' \vee
(q\branap p'\wedge q \branap p''))$. To apply rule ($\branapt$) and
conclude $q \branap p$ we need to prove $\forall p',p''(p\transtt p'
\transt p''\implies q \branap p'\vee q' \branap p'')$.
Let $p',p''$ be such that $p\transtt p'
\transt p''$. Form the induction hypothesis we have two cases.
\begin{itemize}
\item Case $q'\branap p''$; then $q \branap p'\vee q' \branap p''$, so done.
\item Case $q\branap p'\wedge q \branap p''$; then $q \branap p'$ and
  so $q \branap p'\vee q' \branap p''$ and done.
\end{itemize}
\end{proof}

We first state two simple derivable rules, that are nevertheless
convenient to make explicit for use in further proofs.

\begin{lem}\label{lem.branapI}
The following two derived rules holds for $\sbranap$.
\begin{enumerate}
\item
  \[
\begin{prooftree}
  p\transtt t \quad q\transt q'\quad   q'\sbranap  p \quad  \forall p',p''(p\transtt p' \transt p'' \implies q' \sbranap p''\vee (q\sbranap p' \wedge q\sbranap p'')
  ) 
  \justifies
q\sbranap t
\end{prooftree}
\]
\item
\[\begin{prooftree}
p\transtt t \qquad  q\transa q'\qquad \forall p',p''(p\transtt p' \transa p'' \implies q\sbranap p'\vee
  q' \sbranap p'')
  \justifies
  q\sbranap t
\end{prooftree}
\]
\end{enumerate}
\end{lem}

\begin{proof}
For the proof of (1), assume (a) $p\transtt t$, (b) $q\transt q'$, (c)
$q'\sbranap p$ and (d) $\forall p',p''(p\transtt p' \transt p''
\implies q' \sbranap p''\vee (q\sbranap p' \wedge q\sbranap
p'')$. Then by rule ($\sbranapt$), we find $q\sbranap p$, so we may
assume that (a) $p\transtt t$ is non-empty and we have (e) $p \transtt p'
\transt t$.

We use (e) in (d), taking $t$ for $p''$ and find that $q' \sbranap
t$ or $q\sbranap p' \wedge q\sbranap t$. In the latter case  we have $q\sbranap t$ and we are done. In case $q' \sbranap t$, to prove $q\sbranap t$, we apply rule
($\sbranapt$). We need to show that $\forall t',t''(t\transtt t'
\transt t'' \implies q' \sbranap t''\vee (q\sbranap t' \wedge
q\sbranap t'')$, which follows from $p \transtt p' \transt t$ and (d).

The proof of (2) is similar, but slightly simpler.

\end{proof}

\begin{lem}\label{lem.branbisII}
The following two derived rules holds for $\sbranbis$ (and as a matter of
fact they hold for any semi-branching bisimulation relation).
\begin{enumerate}
  \weg{
  \item
  \[
\begin{prooftree}
  q \transtt q'\transt q''\qquad q\sbranbis p
  \justifies
  q''\sbranbis  p \vee \exists p',p''(p\transtt p' \transt p'' \wedge q'\sbranbis p'\wedge
  q'' \sbranbis p'')
\end{prooftree}
\]}
  \item
  \[
\begin{prooftree}
  q_0 \transtt q\transt q'\qquad q_0\sbranbis p
  \justifies
  q'\sbranbis  p \vee \exists p',p''(p\transtt p' \transt p'' \wedge q'\sbranbis p''\wedge
  (q \sbranbis p' \vee q \sbranbis p''))
\end{prooftree}
\]

\item
\[\begin{prooftree}
  q_0 \transtt q\transa q'\qquad q_0\sbranbis p
  \justifies
\exists p',p''(p\transtt p' \transa p'' \wedge q\sbranbis p'\wedge
  q' \sbranbis p'')
\end{prooftree}
\]
\end{enumerate}
\end{lem}

\begin{proof}
Assuming $q_0 \transtt q'$ has the shape $q_0 \transt q_1 \ldots
\transt q_n=q$, the proof proceeds by induction on $n$. We only treat (1), because (2) is similar (but slightly simpler).
  \begin{itemize}
\item ($n=0$) We need to show that the following holds
  \[
  \begin{prooftree}
  q_0 \transt q'\qquad q_0\sbranbis p
  \justifies
  q'\sbranbis  p \vee \exists p',p''(p\transtt p' \transt p'' \wedge q' \sbranbis p'' \wedge(q_0\sbranbis p'\vee q_0\sbranbis p'')
  )
\end{prooftree}
  \]
  which is immediate by ($\sbranbist$)
\item ($n>0$) We need to show that the following holds
  \[
  \begin{prooftree}
  q_0 \transtt q_1 \transt q\transt q'\qquad q_0\sbranbis p
  \justifies
  q'\sbranbis  p \vee \exists p',p''(p\transtt p' \transt p'' \wedge q' \sbranbis p'' \wedge(q\sbranbis p'\vee q\sbranbis p''))
\end{prooftree}
  \]
  By the induction hypothesis, we find
  $$  q\sbranbis  p \vee \exists p',p''(p\transtt p' \transt p'' \wedge q \sbranbis p'' \wedge(q_1\sbranbis p'\vee q_1\sbranbis p'')).$$
\begin{itemize}
\item  In case $q\sbranbis p$, we apply rule ($\sbranbist$) and conclude
  $q'\sbranbis p \vee \exists p',p''(p\transtt p' \transt p'' \wedge
  q' \sbranbis p'' \wedge(q\sbranbis p'\vee q\sbranbis p''))$ and we are done.
\item In the other case, consider the $p', p''$ for which  $p\transtt p' \transt p''$, $q \sbranbis p''$ and $q_1\sbranbis p'\vee q_1\sbranbis p''$. We have $q \transt q'$ and $q \sbranbis p''$, so by rule ($\sbranbist$) we derive  $q'\sbranbis p'' \vee \exists p_1,p_2(p''\transtt p_1 \transt p_2 \wedge
  q' \sbranbis p_2 \wedge(q\sbranbis p_1\vee q\sbranbis p_2))$.
\begin{itemize}
\item In case $q'\sbranbis p''$, we have $q'\sbranbis p'' \wedge
  (q\sbranbis p'\vee q\sbranbis p'')$ and we are done.
\item In case $\exists p_1,p_2(p''\transtt p_1 \transt p_2 \wedge q'
  \sbranbis p_2 \wedge(q\sbranbis p_1\vee q\sbranbis p_2))$ we also
  have $\exists p_1,p_2(p\transtt p_1 \transt p_2 \wedge q' \sbranbis
  p_2 \wedge(q\sbranbis p_1\vee q\sbranbis p_2))$ and we are done.
\end{itemize}
\end{itemize}
  \end{itemize}

\end{proof}

\begin{cor}\label{cor.branapII}
The following two derived rules holds for $\sbranap$.
\begin{enumerate}
\item
  \[
\begin{prooftree}
  q_0 \transtt q\transt q'\qquad   q'\sbranap  p \qquad  \forall p',p''(p\transtt p' \transt p'' \implies q' \sbranap p''\vee (q\sbranap p' \wedge q\sbranap p'')
  ) 
  \justifies
q_0\sbranap p
\end{prooftree}
\]
\item
\[\begin{prooftree}
  q_0 \transtt q\transa q'\qquad \forall p',p''(p\transtt p' \transa p'' \implies q\sbranap p'\vee
  q' \sbranap p'')
  \justifies
  q_0\sbranap p
\end{prooftree}
\]
\end{enumerate}
\end{cor}

\begin{proof}
Immediately from Lemma \ref{lem.branbisII} by taking the complement. 
\end{proof}

In the literature on branching bisimulation, the ``stuttering
property'' refers to the following property for a relation $R$, that we depict as a rule
here.
\begin{equation}
  \label{stutprop}
  \begin{prooftree}
  r \transt r_1 \transt\ldots \transt r_n \trans t \quad (n\geq0)
  \qquad R(r,p) \qquad R(t,p)
  \justifies
   \forall i (1\leq i \leq n) R(r_i,p)
  \end{prooftree}
\end{equation}
So, if in a $\tau$-path, the first and the last state are bisimilar with
$p$, then all states in between are bisimilar with $p$.
In~\cite{GlabbeekWeijland} (and also in other papers), the stuttering
property is proved for $\branbis$.  We cast this property in terms of
apartness.

\begin{defi}\label{def:stuttering}
A relation $Q$ satisfies the {\em apartness stuttering property\/} if
the following rule holds for $Q$.
\[\begin{prooftree}
  r \transtt q \transtt t  \qquad Q(q,p)
  \justifies
  Q(r,p) \vee Q(t,p)
  \using{\stut}
\end{prooftree}\]
\end{defi}

The equivalence between $Q$ being an apartness stuttering property and
$\neg Q$ satisfying the stuttering property of \ref{stutprop} should
be clear. Another way of phrasing the stuttering property for
bisimulations, e.g.\ in \cite{DeNicolaVaandrager}, is as follows.
\begin{equation}
  \label{stutpropNV}
  \begin{prooftree}
  r \transt r_1 \transt\ldots \transt r_n \trans t \quad (n\geq0)
  \qquad R(r,t)
  \justifies
   \forall i (1\leq i \leq n) R(r_0,r_i)
  \end{prooftree}
\end{equation}
The apartness variation of this property is
\begin{equation}
  \label{stutpropNVap}
  \begin{prooftree}
  r \transtt s\transtt t \quad
  \qquad Q(r,s)
  \justifies
Q(r,t)
  \end{prooftree}
\end{equation}
Property \ref{stutpropNVap} follows easily from the apartness stuttering property of
Definition \ref{def:stuttering}, using irreflexivity of $Q$.

\begin{lem}\label{lem.apstut}
The relation $\sbranap$ (semi-branching apartness) satisfies the
apartness stuttering property:
\[\begin{prooftree}
  r \transtt q \transtt t  \qquad q\sbranap p
  \justifies
  r \sbranap p \vee t\sbranap p
  \using{\stut}
\end{prooftree}\]
\end{lem}

\begin{proof}
By induction on the proof of $q\sbranap p$. There are four cases to
consider: either $q\sbranap p$ was derived by rule $(\sbranapt)$ or
$(\branapa)$, or $p\sbranap q$ was derived by rule $(\sbranapt)$ or
$(\branapa)$, and then $q\sbranap p$ was derived by symmetry
($\symm$).
\begin{itemize}
\item Case $q\sbranap p$ was derived using rule ($\sbranapt$). 
So we have
\[
  \begin{prooftree}
    q \transt q' \qquad  q'\sbranap p
    \qquad \forall p',p''(p\transtt p' \transt p''\implies q'\sbranap p'' \vee (q\sbranap p'\wedge q \sbranap p''))
    \justifies
    q \sbranap p
    \using{\sbranapt}
  \end{prooftree}
\]
Then we conclude $r \sbranap p$ using Corollary \ref{cor.branapII}
(1), and so $r \sbranap p \vee t\sbranap p$.

\item Case $q\sbranap p$ was derived using rule ($\branapa$). 
So we have
\[
  \begin{prooftree}
    q \transa q' \qquad \forall p',p''(p\transtt p' \transa p''\implies q \sbranap p'\vee q' \sbranap p'')
    \justifies
    q \sbranap p
        \using{\branapa}
  \end{prooftree}
\] 
Then we conclude $r \sbranap p$ using Corollary \ref{cor.branapII}
(2), and so $  r \sbranap p \vee t\sbranap p$.

\item Case $p\sbranap q$ was derived using rule ($\sbranapt$). 
So we have
\[
  \begin{prooftree}
    p \transt p' \qquad  p'\sbranap q
    \qquad \forall q',q''(q\transtt q' \transt q''\implies p'\sbranap q'' \vee (p\sbranap q'\wedge p \sbranap q''))
    \justifies
    p \sbranap q
    \using{\sbranapt}
  \end{prooftree}
\]
Then we conclude $p \sbranap t$ using Lemma \ref{lem.branapI} (1), and
so $r \sbranap p \vee t\sbranap p$.

\item Case $p\sbranap q$ was derived using rule ($\branapa$). 
So we have
\[
  \begin{prooftree}
    p \transa p' \qquad \forall q',q''(q\transtt q' \transa q''\implies p \sbranap q'\vee p' \sbranap q'')
    \justifies
    p \sbranap q
        \using{\branapa}
  \end{prooftree}
\] 
Then we conclude $p \sbranap t$ using Lemma \ref{lem.branapI} (2), and
so $  r \sbranap p \vee t\sbranap p$.
\end{itemize}
\end{proof}

\begin{lem}\label{lem.bransbran}
Branching apartness is included in semi-branching apartness and thereby
the two notions coincide: $\branap\quad =\quad \sbranap$.
\end{lem}

\begin{proof}
We prove $q \branap p \implies q\sbranap p$ by induction on the
derivation of $q \branap p$, using the apartness stuttering
property. We conclude $\branap\quad =\quad \sbranap$ using Lemma
\ref{lem.sbranbran}.

For the induction we only treat the case where $q\branap p$ has been
derived using the rule ($\branapt$), as the other cases are
immediate. So assume we have the following.
  \[
  \begin{prooftree}
    q \transt q' \qquad  q' \branap p \qquad \forall p',p''(p\transtt p' \transt p''\implies q \branap p'\vee q' \branap p'')
    \justifies
    q \branap p
    \using{\branapt}
  \end{prooftree}
\]
So we have $q \transt q'$ and by induction hypothesis we have (a) $q'
\sbranap p$ and (b) $\forall p',p''(p\transtt p' \transt p''\implies q
\sbranap p'\vee q' \sbranap p'')$.  To be able to apply the rule
($\sbranapt$) to conclude $q \sbranap p$, we need to prove
$$\forall p',p''(p\transtt p' \transt p''\implies q' \sbranap p'' \vee
(q \sbranap p'\wedge q \sbranap p'')).$$
Let $p', p''$ be such that $p\transtt p' \transt p''$. Using (b) we
have two cases.
\begin{itemize}
\item Case $q \sbranap p'$. Then, by $p\transtt p' \transt p''$ and
  the stuttering property (Lemma \ref{lem.apstut}), we have $q
  \sbranap p \vee q \sbranap p''$. In case $q \sbranap p$, we are
  done, because that's exactly what we had to prove in the end; in
  case $q \sbranap p''$ we have $q' \sbranap p'' \vee (q \sbranap
  p'\wedge q \sbranap p'')$ and we are done.
\item Case $q' \sbranap p''$. Then $q' \sbranap p'' \vee (q \sbranap
  p'\wedge q \sbranap p'')$ and we are done.
\end{itemize}

\end{proof}

As a consequence of this Lemma, Corollary \ref{cor.branapII} and
Lemma \ref{lem.branapI} also apply to branching apartness, $\branap$.

\begin{lem}\label{lem.branapcotrans}
  The relation $\branap$ is co-transitive: for all $q,p,r$: if $q
  \branap p$, then $q \branap r \vee r \branap p$.
\end{lem}

\begin{proof}
We prove $q\branap p \implies \forall r (q\branap r \vee r \branap p)$
by induction on the derivation of $q\branap p$, using the properties
we have proved before about $\branap$ and $\sbranap$.
\begin{itemize}
  \item Case $q \branap p$ was derived using ($\branapt$).
  \[
  \begin{prooftree}
    q \transt q' \qquad  q' \branap p \qquad \forall p',p''(p\transtt p' \transt p''\implies q \branap p'\vee q' \branap p'')
    \justifies
    q \branap p
    \using{\branapt}
  \end{prooftree}
  \]
  Let $r$ be a state. If (a) $q'\branap r$ and (b) $\forall
  r',r''(r\transtt r' \transt r''\implies q \branap r'\vee q' \branap
  r'')$, then $q\branap r$ and we are done. Otherwise, $\neg(q'\branap
  r)$ or
  $$\exists r',r''(r\transtt r' \transt r''\wedge \neg(q \branap r')
  \wedge \neg(q' \branap r'')).$$
  If $\neg(q'\branap r)$, we apply
  induction on $q' \branap p$ to derive $q'\branap r \vee r\branap p$, from which
  we conclude $r\branap p$ and we are done.

In the other case
we consider $r',r''$ with
(d) $\neg(q \branap r')$ and (e) $\neg(q' \branap r'')$. We will prove
that $r\branap p$. Let $p',p''$ be such that $p\transtt p' \transt
p''$. (If there are no such $p', p''$, then $r\branap p$ due to
Corollary \ref{cor.branapII} (1) and the fact that $r'' \branap p$,
which follows from induction on $q' \branap p$, which yields
$q'\branap r'' \vee r''\branap p$, but we know $\neg(q' \branap
r'')$ from (e).) Then $q \branap p'\vee q' \branap p''$.
\begin{itemize}
\item Case $q \branap p'$. Then by induction $q \branap r' \vee
  p'\branap r'$, so $p'\branap r'$ by (d) and so $p'\branap r' \vee
  p''\branap r''$.
\item Case $q' \branap p''$. Then by induction $q' \branap r'' \vee
  p''\branap r''$, so $p''\branap r''$ by (e) and so $p'\branap r'
  \vee p''\branap r''$.
\end{itemize}
So $r'\branap p$ and we apply
Corollary \ref{cor.branapII} (1), to conclude $r\branap p$.
\item Case $q \branap p$ was derived using ($\branapa$).
\[
  \begin{prooftree}
    q \transa q' \qquad \forall p',p''(p\transtt p' \transa p''\implies q \branap p'\vee q' \branap p'')
    \justifies
    q \branap p
        \using{\branapa}
  \end{prooftree}
\] 
  Let $r$ be a state. If $\forall
  r',r''(r\transtt r' \transa r''\implies q \branap r'\vee q' \branap
  r'')$, then $q\branap r$ and we are done. Otherwise
  $$\exists r',r''(r\transtt r' \transa r''\wedge \neg(q \branap r')
  \wedge \neg(q' \branap r'')).$$
Consider $r',r''$ with (d) $\neg(q \branap r')$ and
(e) $\neg(q' \branap r''))$. We will prove that $r\branap p$. Let
$p',p''$ be such that $p\transtt p' \transa p''$. (If there are no such $p', p''$, then $r\branap p$ due to Corollary \ref{cor.branapII} (2).) Then $q \branap
p'\vee q' \branap p''$.
\begin{itemize}
\item Case $q \branap p'$. Then by induction $q \branap r' \vee
  p'\branap r'$, so $p'\branap r'$ by (d) and so $p'\branap r' \vee
  p''\branap r''$.
\item Case $q' \branap p''$. Then by induction $q' \branap r'' \vee
  p''\branap r''$, so $p''\branap r''$ by (e) and so $p'\branap r'
  \vee p''\branap r''$.
\end{itemize}
So $r'\branap p$ and we apply
Corollary \ref{cor.branapII} (2), to conclude $r\branap p$.
\end{itemize}

\end{proof}
      
The co-transitivity is the crucial property for showing that $\branap$
is a proper apartness relation.

\begin{thm}
  \label{thm.apartisapart}
  The relation $\branap$ is a
  proper apartness relation (in the sense of Definition
  \ref{def:properapart}).
\end{thm}

\begin{proof}
We need to verify irreflexivity, symmetry and
co-transitivity. Symmetry is built in and co-transitivity has been
proved in Lemma \ref{lem.branapcotrans}. For irreflexivity, consider
the shortest derivation of $q\branap q$ (for some $q$). If this is
derived using rule $\branapa$, we have $q\transa q'$ and $q\branap q
\vee q'\branap q'$, which means that there is a shorter derivation of
a reflexivity, contradiction.  If this is derived using rule
$\branapt$, we have $q\transt q'$ and $q\branap q \vee q'\branap q'$,
which again means that there is a shorter derivation of a reflexivity,
contradiction. So there is no derivation of $q\branap q$ for any $q$.
\end{proof}

\begin{cor}
\label{lem.equiv}
The relation $\branbis$ is an equivalence relation.
\end{cor}

\begin{proof}
Immediately from Theorem \ref{thm.apartisapart} using the fact that
$\branbis$ is the complement of $\branap$.
\end{proof}

\subsection{Using branching apartness}\label{sec:usingapartness}
Further research has to establish whether the notion of apartness is
really useful in the study and analysis of labelled transition
systems. In the previous section we have shown how to use apartness in
the meta-theory of branching bisimulation to give some new proofs for
known properties. In follow up research we would like to analyze
well-known algorithms for checking branching bisimulation, as
in~\cite{Jansenetal}, and possibly develop variations on those
algorithms.  One way to decide branching bisimilarity of states in a
finite LTS is by deciding branching apartness. In the present section,
we give some ideas of what an algorithm for deciding branching
apartness could look like and we also give some variations of the
rules for branching apartness, also combined with branching
bisimulation that might provide useful. In the end of this section, we
briefly mention {\em rooted branching apartness\/} as the complement
of {\em rooted branching bisimulation}. Rooted branching bisimulation
is a congruence \cite{GlabbeekWeijland,Fokkink}, while branching
bisimulation is not. For apartness this means that operations are
strongly extensional with respect to rooted branching apartness, while
they are not with respect to branching apartness.

An obvious algorithm to decide $q \branap p$ is by trying to find a
derivation of $q\branap p$ in a structured way and concluding that
$q\branbis p$ holds in case such a derivation cannot be found. It may
look as if, for LTSs with loops, this could lead to an infinite search
process. But this can be avoided if we look for a {\em shortest
  derivation\/} and keep track of goals that we have already
encountered. If we encounter the goal again, we can conclude it is not
provable. Also, some of the goals will be disjunctions of apartness
assertions, like $q \branap p'\vee q' \branap p''$. In that case we
will search for a proof of $q \branap p'$ and for a proof of $q'
\branap p''$, in parallel, and we conclude as soon as we have found a
proof of one of them. To clarify this point a bit better we show two
pairs of LTSs with loops and how a proof of branching apartness is
found for the first pair, and a proof of branching bisimilarity for the
second pair.

\begin{exa}\label{exa.infLTS}
We give 4 LTSs with loops.\\  
\begin{tabular}[t]{ccccc}
\begin{tikzpicture}[>=stealth,node distance=1.5cm,auto]
    \node  (p0)                        {$p_0$};
    \node  (p1) [below of = p0]        {$p_1$};
    \node  (p2) [right of = p0]        {$p_2$};
    \node  (p3) [below of = p2]        {$p_3$};

    \path[->]
    (p0) edge                   node[swap] {$d$}    (p1)
         edge                   node {$\tau$}       (p2)
    (p1) edge [bend right]     node[swap] {$c$}    (p0)
         edge [bend left=90]     node {$e$}          (p0)
    (p2) edge      node[swap] {$d$}    (p3)
    (p3) edge [bend right]      node[swap] {$c$}          (p2);
\end{tikzpicture}
&
\qquad
\begin{tikzpicture}[>=stealth,node distance=1.5cm,auto]
    \node  (q0)                        {$q_0$};
    \node  (q1) [below of = q0]        {$q_1$};
    \node  (q2) [right of = q0]        {$q_2$};
    \node  (q3) [below of = q2]        {$q_3$};
    \node  (q4) [right of = q3]        {$q_4$};
    \path[->]
    (q0) edge                   node[swap] {$d$}    (q1)
         edge                   node {$d$}       (q2)
    (q1) edge [bend right]      node[swap] {$c$}    (q0)
         edge [bend left=90]    node {$e$}          (q0)
    (q2) edge                   node[swap] {$c$}    (q3)
    (q3) edge [bend right]      node[swap] {$d$}    (q4)
    (q4) edge [bend right]      node[swap] {$c$}    (q3);
\end{tikzpicture}
&\qquad

\qquad&
\begin{tikzpicture}[>=stealth,node distance=1.5cm,auto]
    \node  (q)                        {$q$};
    \node  (q') [below of = q]        {$q'$};

    \path[->]
    (q)  edge  [bend right]      node[swap] {$d$}       (q')
    (q') edge  [bend right]      node[swap] {$d$}       (q);
\end{tikzpicture}
&
\qquad

\begin{tikzpicture}[>=stealth,node distance=1.5cm,auto]
    \node  (p)                        {$p$};

    \path[->]
    (p)  edge   [loop below=90]    node[swap] {$d$}       (p);
\end{tikzpicture}
    \end{tabular}

\noindent In the first two LTSs, we have $q_0 \branap p_0$, which is established by the derivation below.
{\small
\[
\begin{prooftree}
q_0 \trans_d q_2 
\quad
 \begin{prooftree}
  \begin{prooftree}
   \begin{prooftree}
     \begin{prooftree}
      p_1 \trans_e p_0
      \justifies
      p_1 \branap q_2
     \end{prooftree}
       \justifies
      q_2 \branap p_1
     \end{prooftree}
  \justifies
        q_0 \branap p_0 \vee q_2 \branap p_1
  \end{prooftree}
\qquad
 \begin{prooftree}
   \begin{prooftree}
        q_0 \trans_d q_1 \quad
      \begin{prooftree}
        \begin{prooftree}
          \begin{prooftree}
           q_1 \trans_e q_0
          \justifies
           q_1 \branap p_3
          \end{prooftree}
        \justifies 
          q_0 \branap p_2 \vee q_1 \branap p_3
        \end{prooftree}
      \justifies
        \forall p',p''(p_2\transtt p' \trans_d p''\implies q_0 \branap p'\vee q_1 \branap p'')
      \end{prooftree}
    \justifies
      q_0 \branap p_2
    \end{prooftree}
   \justifies
        q_0 \branap p_2 \vee q_2 \branap p_3
  \end{prooftree}
  \justifies
  \forall p',p''(p_0\transtt p' \trans_d p''\implies q_0 \branap p'\vee q_2 \branap p'')
 \end{prooftree}
  \justifies
  q_0\branap p_0
\end{prooftree}\]
}

Observe that an algorithm would have to go through all possible
$d$-steps from $q_0$ and ``replay'' them from $p_0$. We have chosen
the ``successful'' $d$-step that leads to a derivation of $q_0\branap
p_0$. Similarly in proving $ q_0\branap p_2$, we have chosen the
successful $d$-step, $q_0\trans_d q_1$. When proving a disjunction, an
algorithm would have to try to prove both parts of the disjunction in
parallel. We have only shown the successful one.

For the third and fourth LTS, we have $q \branbis p$, so we want to show
that $\neg (q\branap p)$. This is achieved by trying to find the
shortest derivation and observing there is none. This search leads to
the following derivation.
\[ \begin{prooftree}
  q \trans_d  q' \quad
  \begin{prooftree}
    \begin{prooftree}
      q' \trans_d  q \quad
      \begin{prooftree}
        \mbox{fail}
        \justifies
            q'\branap p \vee q\branap p
          \end{prooftree}
      \justifies
      q'\branap p
    \end{prooftree}
    \justifies
    q\branap p \vee q'\branap p
  \end{prooftree}
  \justifies
  q\branap p
\end{prooftree}
  \]
  Note that this is the complete search tree for a derivation of
  $q\branap p$, where we have stopped at a branch as soon as we find a
  goal that we have already encountered. Therefore we fail at the goal
  $q'\branap p \vee q\branap p$, because both $q'\branap p$ and
  $q\branap p$ have already been encountered. Conclusion:
  $\neg(q\branap p)$, that is $q\branbis p$.
\end{exa}

We now look into some variations on the rules for branching apartness.

\begin{lem}\label{lem.branapalt}
  The following alternative $\branapa$-rule\footnote{Thanks to David N.\ Jansen for suggesting
    this rule} is sound for $\branap$.
 \[
  \begin{prooftree}
    q \transa q' \qquad \forall p',p''(p\transtt p' \transa p''\implies p \branap p'\vee q' \branap p'')
    \justifies
    q \branap p
        \using{\branapaalt}
  \end{prooftree}
\] 
\end{lem}

\begin{proof}
\weg{First we prove rule ($\branapaalt$) from rule ($\branapa$), soundness.}
Assume we have $q \transa q'$ and $\forall p',p''(p\transtt p' \transa
p''\implies p \branap p'\vee q' \branap p'')$. We need to prove $q
\branap p$. Suppose $\neg q \branap p$. We want to apply the original
$\branapa$-rule, so we need to prove the hypothesis to that rule,
which is $\forall p',p''(p\transtt p' \transa p''\implies q \branap
p'\vee q' \branap p'')$. Let $p', p''$ be such that $p\transtt p'
\transa p''$. Then $p \branap p'\vee q' \branap p''$.
\begin{itemize}
\item Case $p \branap p'$. Then $q \branap p'\vee q \branap p$ by
co-transitivity. We know from our assumption that $\neg q \branap p$,
so $q \branap p'$ and so $q \branap p'\vee q' \branap p''$ and done.
\item Case $q' \branap p''$. Then $q \branap p'\vee q' \branap p''$ and done.
So we can apply the original $\branapa$-rule and conclude $q \branap
p$. This contradicts our assumption $\neg q \branap p$, so we conclude
$q \branap p$.
\qedhere
\end{itemize}
\weg{
  Now we prove rule ($\branapa$) from rule ($\branapaalt$), completeness.
Assume we have $q \transa q'$ and $\forall p',p''(p\transtt p' \transa
p''\implies q \branap p'\vee q' \branap p'')$. We need to prove $q
\branap p$. Suppose $\neg q \branap p$. We want to apply the
$\branapaalt$-rule, so we need to prove the hypothesis to that rule,
which is $\forall p',p''(p\transtt p' \transa p''\implies p \branap
p'\vee q' \branap p'')$. Let $p', p''$ be such that $p\transtt p'
\transa p''$. Then $q \branap p'\vee q' \branap p''$.
\begin{itemize}
\item Case $q \branap p'$. Then $q \branap p\vee p \branap p'$ by
co-transitivity. We know from our assumption that $\neg q \branap p$,
so $p \branap p'$ and so $p \branap p'\vee q' \branap p''$ and done.
\item Case $q' \branap p''$. Then $q \branap p'\vee q' \branap p''$ and done.
So we can apply the $\branapaalt$-rule and conclude $q \branap
p$. This contradicts our assumption $\neg q \branap p$, so we conclude
$q \branap p$.
\end{itemize}
}
\end{proof}

We conjecture that the rule ($\branapaalt$) is also complete for
proving $\branap$, that is: if we replace rule ($\branapa$) with rule
($\branapaalt$) we can derive the same apartness judgments. If we
write $\branapA$ for the system with rule ($\branapa$) replaced by
rule ($\branapaalt$), Lemma \ref{lem.branapalt} states that $q
\branapA p \implies q \branap p$. For the proof of completeness, $q
\branap p \implies q \branapA p$, it seems we need to prove
co-transitivity of $\branapA$ first.

\weg{
\begin{lem}
If $Q$ is a co-transitive semi-branching bisimulation that satisfies
the apartness stuttering property, it is a branching bisimulation.
\end{lem}

\begin{proof}
  Let $Q$ be a co-transitive semi-branching bisimulation that
  satisfies the apartness stuttering property. We need to prove that
  $Q$ satisfies the rule $\branapt$:
  \[
  \begin{prooftree}
    q \transt q' \qquad  Q(q',p) \qquad \forall p',p''(p\transtt p' \transt p''\implies Q(q,p')\vee Q(q',p''))
    \justifies
    Q(q,p)
    \using{\branapt}
  \end{prooftree}
\]
So suppose $q \transt q'$, $Q(q',p)$ (1) and $\forall p',p''(p\transtt
p' \transt p''\implies Q(q,p')\vee Q(q',p''))$ (2), and assume $\neg Q(q,p)$.
We want to apply the $\sbranapt$-rule to conclude $Q(q,p)$, so we are done if we
prove the ``missing'' condition for the $\sbranapt$-rule:
\[\forall p'(p\transtt p'\implies Q(q,p')\vee Q(q',p')).\]

\noindent So let $p'$ be such that $p\transtt p'$.

If $p'=p$ and $p\transtt p'$ is empty, we are done.

If $p\transt p'$, then we have (using (2)): $Q(q,p)\vee Q(q',p')$. The
first contradicts our assumption and the second implies $Q(q,p')\vee
Q(q',p')$.

If $p\transt p_0 \transt p'$, then (using (2)) we have : $Q(q,p_0)\vee
Q(q',p')$. The second implies $Q(q,p')\vee Q(q',p')$. For the first,
we apply the stuttering property to derive $Q(q,p) \vee Q(q,p')$. Now
$Q(q,p)$ contradicts our assumption, so we have $Q(q,p')$ from which
we conclude $Q(q,p')\vee Q(q',p')$.
\end{proof}
}
\smallskip

Using the notion of apartness, we can also add some rules that combine
apartness and bisimulation and that may be useful in analyzing or
developing new algorithms for checking branching bisimulation, as
in~\cite{Jansenetal}.

\begin{lem}\label{lem.soundrules-apt}
  The following two rules are sound for proving branching apartness $\branap$.
  \[
  \begin{prooftree}
    q \transt q' \qquad  q' \branap p \qquad \forall p',p''(p\transtt p' \transt p''\wedge q' \branbis p''\implies q \branap p')
    \justifies
    q \branap p
    \using{\branapt^{\branbis}}
  \end{prooftree}
\]

\[
  \begin{prooftree}
    q \transa q' \qquad \forall p',p''(p\transtt p' \transa p'' \wedge q' \branbis p''\implies q \branap p')
    \justifies
    q \branap p
        \using{\branapa^{\branbis}}
  \end{prooftree}
\] 
\end{lem}

\begin{proof}
The proof is immediate from the fact that $\branap = \neg \branbis$ and Corollary \ref{cor.branap}.
\end{proof}

In the literature, the rules concerning  bisimulation are often depicted in a diagram for better memorization.
The two rules above can be depicted as follows.

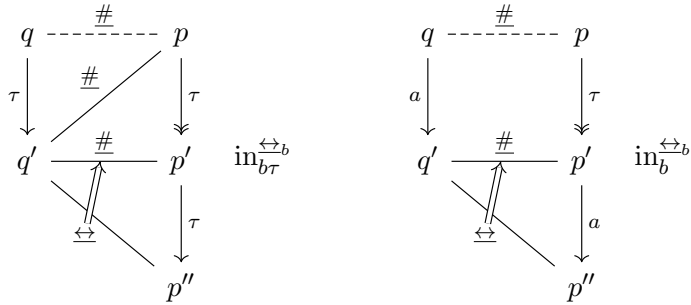
\begin{tabular}[t]{rlcrl}
  \begin{tikzcd}[row sep=large,column sep=large]
    q\arrow[r, dash, dashed, "\apt"]\arrow[d, "\tau"'] & p\arrow[d, two heads, "\tau"] \\
    q'\arrow[ru, dash, "\apt"]\arrow[r, dash, "\apt"{name=U}]\arrow[rd, dash, "\bis"'{name=V}]& p' \arrow[d, "\tau"] \\
    & p''
    \arrow[Rightarrow,from=V,to=U]
  \end{tikzcd}
&  $\branapt^{\branbis}$
&\qquad\qquad&
  \begin{tikzcd}[row sep=large,column sep=large]
    q\arrow[r, dash, dashed, "\apt"]\arrow[d, "a"'] & p\arrow[d, two heads, "\tau"] \\
    q'\arrow[r, dash, "\apt"{name=U}]\arrow[rd, dash, "\bis"'{name=V}]& p' \arrow[d, "a"] \\
    & p''
    \arrow[Rightarrow,from=V,to=U]
  \end{tikzcd}
  & $\branapa^{\branbis}$
    \end{tabular}

\noindent On the left: Suppose that $q' \branap p$ and for all $p',p''$ with
$p\transtt p' \transt p''$, if $q' \branbis p''$, then $q \branap
p'$. Then (in dashes): $q \branap p$.

Similarly, we introduce an adapted rule for branching bisimulation.

\begin{lem}\label{lem.soundrules-bis}
  The following rule is sound for branching bisimulations. If $R$ is a
  branching bisimulation, then the following rule (and its symmetric
  variant) holds.
\[
\begin{prooftree}
  q \transt q'\qquad R(q,p)  \qquad  q'\branap p
  \justifies
\exists p',p''(p\transtt p' \transt p'' \wedge R(q,p')\wedge
  R(q',p''))
\end{prooftree}
\]
\end{lem}

\begin{proof}
The proof is immediate from the definition of branching bisimulation
(Definition \ref{def.weakbran}) and the fact that, if $R$ is a
branching bisimulation, then $q'\branap p \implies \neg R(q',p)$.
\end{proof}

The notion of rooted branching bisimulation has been introduced to
recover the failure of the congruence property for branching
bisimulation \cite{GlabbeekWeijland, BaetenBastenReniers,Fokkink}. If
two labelled transition systems are branching bisimilar
and one composes them both with a third one, one expects the newly obtained LTSs to
be branching bisimilar again. But this is not the case. Of course,
this also depends on the notion of composition and therefore this
problem is usually cast in terms of process terms and a process
operator $f$, where congruence means that if $q_1 \branbis p_1$ and
$q_2\branbis p_2$, then $f(q_1,q_2) \branbis f(p_1, p_2)$. Here, the
bisimulation equivalence should be understood as being between the
interpretation of the process terms as LTSs. A main example of
non-congruence arises from the non-deterministic choice operator
$+$. A process $q+p$ can non-deterministically choose for $q$ and do
a step in $q$ or for $p$ and do a step in $p$. The LTS for $q+p$
arises from joining the LTS for $q$ with the one for $p$ at the root
node. This leads to the non-congruence exemplified
below.

\begin{exa}\label{exa:noncong}
For the two LTSs on the left we have $q_0 \branbis p_0$. If we compose
them via non-deterministic choice with an LTS that does just a
$c$-step, we get the two LTSs on the right, for which we have $\neg(q_0
\branbis p_0)$.\\
  
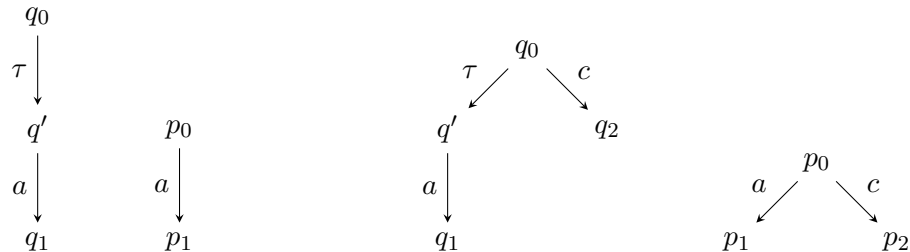
\begin{tabular}[t]{ccccc}
\begin{tikzpicture}[>=stealth,node distance=1.5cm,auto]
    \node    (q0)                         {$q_0$};
    \node    (q') [below of = q0]         {$q'$};
    \node    (q1) [below of = q']         {$q_1$};

    \path[->]
        (q0) edge                    node[swap] {$\tau$}  (q')
        (q') edge                    node[swap] {$a$}     (q1);
\end{tikzpicture}
&
\qquad
\begin{tikzpicture}[>=stealth,node distance=1.5cm,auto]
    \node    (p0)                         {$p_0$};
    \node    (p1) [below of = p0]         {$p_1$};

    \path[->]
        (p0) edge                    node[swap] {$a$}  (p1);
 
\end{tikzpicture}
\qquad\qquad
&
\qquad\qquad
\begin{tikzpicture}[>=stealth,node distance=1.5cm,auto]
    \node    (q0)                         {$q_0$};
    \node    (q') [below left of = q0]    {$q'$};
    \node    (q1) [below of = q']         {$q_1$};
    \node    (q2) [below right of = q0]   {$q_2$};

    \path[->]
        (q0) edge                    node[swap] {$\tau$}  (q')
             edge                    node {$c$}           (q2)
        (q') edge                    node[swap] {$a$}     (q1);
\end{tikzpicture}
&
\qquad
\begin{tikzpicture}[>=stealth,node distance=1.5cm,auto]
    \node    (p0)                         {$p_0$};
    \node    (p1) [below left of = p0]    {$p_1$};
    \node    (p2) [below right of = p0]   {$p_2$};

    \path[->]
        (p0) edge                    node[swap] {$a$}     (p1)
             edge                    node {$c$}           (p2);
\end{tikzpicture}

    \end{tabular}

For the LTSs on the right, a derivation of $q_0 \branap p_0$ is as follows.
\[
\begin{prooftree}
  q_0 \transt q' \quad
  \begin{prooftree}
    \begin{prooftree}
      p_0 \trans_c p_2
      \justifies
         p_0 \branap q'
    \end{prooftree}
    \justifies
    q' \branap p_0
  \end{prooftree}
  \quad
    \begin{prooftree}
       \mbox{trivial}
      \justifies
    \forall p',p''( p_0 \transtt p'\transt p'' \implies q_0 \branap p' \vee q'\branap p'')
    \end{prooftree}
  \justifies
  q_o\branap p_0
\end{prooftree}
\]
\end{exa}

To remedy the non-congruence, the notion of {\em rooted branching
  bisimulation\/} has been introduced.

\begin{defiC}[\cite{GlabbeekWeijland}]\label{def:rootbran}
  A relation $R$ on an LTS is a {\em rooted branching bisimulation\/} if
the following properties hold.
\begin{itemize}
\item $R$ is a branching bisimulation.
\item $R$ satisfies the rule ($\rbranbisx$), where $x$ ranges over $\Atau$,
  \[
  \begin{prooftree}
q \trans_x q' \qquad R(q,p)
    \justifies
\exists p' ( p \trans_x p' \wedge  q'\branbis  p')
    \using{\rbranbisx}
  \end{prooftree}
  \]
\end{itemize}
The states $q$ and $p$ are {\em rooted branching bisimilar}, notation
$q \rbranbis p$ if there is a rooted branching bisimulation $R$ such
that $R(q,p)$.

As a complement we define that the relation $Q$ on an LTS is a {\em rooted branching apartness\/} if
the following properties hold.
\begin{itemize}
\item $Q$ is a branching apartness.
\item $Q$ satisfies the rule ($\rbranapx$), where $x$ ranges over $\Atau$,
  \[
  \begin{prooftree}
q \trans_x q' \qquad \forall p' ( p \trans_x p' \implies  q'\branap  p')
    \justifies
Q(q,p)
    \using{\rbranapx}
  \end{prooftree}
  \]
\end{itemize}
The states $q$ and $p$ are {\em rooted branching apart}, notation
$q \rbranap p$ if for all rooted branching apartness relations $Q$, we have $Q(q,p)$.
\end{defiC}

It is easy to see that $q \rbranbis p$ if and only if $\neg(q\branap
p)$ and that $q\branap p$ if and only if this is derivable using the
rules ($\branap$), ($\branapt$) and ($\rbranapx$).

\begin{exa}\label{exa:rootbran}
As a continuation of Example \ref{exa:noncong}, we can now show that
both for the two LTSs on the right and also for the two LTSs on the
left, we have $q_0 \rbranap p_0$. For the pair on the left, the
derivation is simply as follows.
  \[
  \begin{prooftree}
    q_0 \transt q' \qquad
    \begin{prooftree}
      \mbox{trivial}
      \justifies
    \forall p' ( p_0 \transt p' \implies  q'\branap  p')
    \end{prooftree}
    \justifies
  q_0 \rbranap p_0
  \end{prooftree}
  \]
Thereby, congruence has been regained, by refining branching
bisimulation (or strengthening branching apartness).
\end{exa}

To understand how congruence is regained in general and relate that to
strong extensionality, assume we have an operation $+$ on LTSs that
represents non-deterministic choice, which means to join the
root-nodes of two LTSs. So process terms of the form $q+p$ have the
following behavior as LTS.

\[
\begin{prooftree}
p_1 \transx p_1'
  \justifies
  p_1 + p_2 \transx p_1'
\end{prooftree}
\qquad 
\begin{prooftree}
p_2 \transx p_2'
  \justifies
  p_1 + p_2 \transx p_2'
\end{prooftree}
\]

The interpretation of $q+p$ as an LTS should be understood as the union
of the LTSs arising from $q$ and $p$ that are joined at their root
node. That rooted branching bisimulation is a congruence for the
$+$ operator means that, if $q_1 \rbranbis p_1$ and $q_2 \rbranbis
p_2$, then $(q_1 + q_2) \rbranbis (p_1 + p_2)$. We will look into this
property from the ``apartness side'', where we want $+$ to be strongly
extensional for $\rbranap$: if $(q_1 + q_2) \rbranap (p_1 + p_2)$,
then $q_1 \rbranap p_1 \vee q_2 \rbranap p_2$. We can prove strong
extensionality of $+$ as follows. Assume that $(q_1 + q_2) \rbranap
(p_1 + p_2)$. This means that there is a state $r$ and $x\in\Atau$
with $(q_1 + q_2) \transx r$ and $\forall r' ( (p_1+p_2) \transx r'
\implies r\branap r')$.  The transition $(q_1 + q_2) \transx r$ either
arises from $q_1 \transx r$ or from $q_2\transx r$. In case $q_1
\transx r$ we have the derivation below. In case $q_2\transx r$ we
have a symmetric situation that we don't depict here.
\[
\begin{prooftree}
  \begin{prooftree}
q_1 \transx r
  \justifies
  (q_1 + q_2) \transx r
\end{prooftree}
  \qquad
  \forall r' ( (p_1+p_2) \transx r' \implies r\branap r')
  \justifies
  (q_1 + q_2) \rbranap (p_1 + p_2)
\end{prooftree}
\]
So we have $\forall r' ( (p_1+p_2) \transx r' \implies r\branap r')$,
which implies $\forall r' ( p_1 \transx r' \implies r\branap r')$,
which, together with $q_1 \transx r$, using rule ($\rbranapx$) gives
$q_1 \rbranap p_1$. In case $q_2\transx r$ we
have a symmetric situation, so we can conclude $q_1 \rbranap p_1 \vee q_2 \rbranap p_2$.

\section{Coalgebras and lifting}\label{CoalgLiftSec}

This section recalls some basic facts about the description of
bisimulations on coalgebras in terms of lifting of the functor from
sets to relations.

We write $\Rel$ for the category of binary relations $R \subseteq
X\times X$. A morphism $f \colon (R\subseteq X\times X) \rightarrow
(S\subseteq Y\times Y)$ in $\Rel$ is a function $f \colon X
\rightarrow Y$ between the underlying sets satisfying $(x_{1},x_{2})
\in R \Longrightarrow (f(x_{1}), f(x_{2})) \in S$. This can
equivalently be expressed via the existence of a function $f'$ in:
\begin{equation}
\label{RelMapDiag}
\vcenter{\xymatrix@R-0.5pc{
R\ar@{^(->}[d]_{\tuple{r_{1},r_{2}}}\ar@{-->}[rr]^-{f'} & & 
   S\ar@{^(->}[d]^{\tuple{s_{1},s_{2}}}
\\
X\times X\ar[rr]^-{f\times f} & & Y\times Y
}}
\end{equation}

\noindent We can also describe this situation via an inclusion $R
\subseteq (f\times f)^{-1}(S)$, where $(f\times f)^{-1}(S) =
\set{(x_{1}, x_{2})}{(f(x_{1}), f(x_{2})) \in S}$. There is an obvious
functor $\Rel \rightarrow \Set$ which sends a relation $R\subseteq
X\times X$ to its underlying set $X$. The poset $\Pow(X\times X)$ of
relations on a set $X$ is often called the ``fibre over $X$'', since it
mapped by this functor to $X$.

In this setting we restrict ourselves to (endo)functors $F \colon \Set
\rightarrow \Set$, with associated category $\CoAlg(F)$ of coalgebras.
There is a standard way to ``lift'' such a functor $F$ from $\Set$ to
$\Rel$ in a commuting diagram, as on the left below.
\begin{equation}
\label{RelliftDiag}
\vcenter{\xymatrix@R-0.5pc{
\Rel\ar[rr]^-{\rellift(F)}\ar[d] & & \Rel\ar[d]
& \hspace*{4em} & 
\Rel\ar[rr]^-{\rellift(F)} & & \Rel
\\
\Set\ar[rr]^-{F} & & \Set
& &
\Set\ar[rr]^-{F}\ar[u]^{\Eq} & & \Set\ar[u]_{\Eq}
}}
\end{equation}

\noindent For $R\subseteq X\times X$ one obtains $\rellift(F)(R) \subseteq
F(X)\times F(X)$ via:
\begin{equation}
\label{RelliftEqn}
\begin{array}{rcl}
\rellift(F)(R)
& \coloneqq &
\set{(u_{1},u_{2})}{\exin{w}{F(R)}{F(r_{1})(w) = u_{1}
   \mbox{ and } F(r_{2})(w) = u_{2}}}.
\end{array}
\end{equation}

\noindent Here we write the inclusion map $R \hookrightarrow X\times
X$ as a pair $\tuple{r_{1}, r_{2}} \colon R \rightarrow X\times X$. 

It is not hard to see that $\rellift(F)$ is functorial: for a morphism
$f \colon (R\subseteq X\times X) \rightarrow (S\subseteq Y\times Y)$
in $\Rel$ there is a (unique) map $f'\colon R \rightarrow S$ as
in~\eqref{RelMapDiag}. We show that $F(f)$ is a morphism
$\rellift(f)(R) \rightarrow \rellift(F)(S)$. So let $(u_{1},u_{2}) \in
\rellift(F)(R)$, say via $w\in F(R)$ with $F(r_{i})(w) = u_{i}$.  Then
$w' \coloneqq F(f')(w) \in F(S)$ satisfies $F(s_{i})(w') = F(s_{i}
\after f')(w) = F(f \after r_{i})(w) = F(f)(u_{i})$. This shows that
the pair $(F(f)(u_{1}), F(f)(u_{2}))$ is in $\rellift(F)(S)$.

A \emph{bisimulation} for a coalgebra $c\colon X \rightarrow F(X)$ is
a relation $R\subseteq X\times X$ for which $c$ is a map in the
category $\Rel$ of the form $c\colon R \rightarrow \rellift(F)(R)$.
Thus we may consider the category of coalgebras
$\CoAlg\big(\rellift(F)\big)$ as the category of bisimulations --- for
$F$-coalgebras.

\begin{lem}
\label{RelliftEqLem}
Relation lifting commutes with equality, as in the diagram on the
right in~\eqref{RelliftDiag}, where $\Eq(X) = \set{(x,x)}{x\in X}
\subseteq X\times X$.
\end{lem}

\begin{myproof}
There is an obvious isomorphism $\varphi$ in:
\[ \xymatrix@R-0.5pc{
\Eq(X)\ar@/_2ex/@{^(->}[dr]_-{\tuple{r_{1},r_{2}}}\ar[rr]_-{\cong}^-{\varphi} 
   & & X\ar@/^2ex/[dl]^-{\tuple{\idmap,\idmap}}
\\
& X &
} \]

\noindent In particular, $r_{1} = r_{2}$.

Hence if $(u_{1},u_{2})\in\rellift(F)(\Eq(X))$, say via $w\in
F(\Eq(X))$ with $F(r_{i}) = u_{i}$, then $u_{1} = u_{2}$ since $r_{1}
= r_{2}$. In the other direction, for $u\in F(X)$ we have $(u,u) \in
\rellift(F)(\Eq(X))$ via $w = F(\varphi^{-1})(u) \in F(\Eq(X))$ that
satisfies $F(r_{i})(w) = F(\idmap)(u) = u$. \QED
\end{myproof}

For (Kripke) polynomial functors the generic form of relation
lifting~\eqref{RelliftEqn} specializes to well-known formulas,
see~\cite{Jacobs16} for details.

\begin{lem}
\label{RelliftPolynomialLem}
Relation lifting satisfies:
\begin{enumerate}
\item $\rellift(\idmap) = \idmap$;

\item $\rellift(A) = \Eq(A)$, where $A$ on the left is the
  constant-$A$ functor;

\item $\rellift(F_{1}\times F_{2})(R) = 
   \set{(u,v)}{(\pi_{1}u,\pi_{1}v)\in\rellift(F_{1})(R) \mbox{ and }
    (\pi_{2}u,\pi_{2}v)\in\rellift(F_{1})(R)}$;

\item $\rellift(F_{1} + F_{2})(R) = 
   \set{(\kappa_{1}u,\kappa_{1}v)}{(u,v)\in\rellift(F_{1})(R)} \,\cup\, 
   \set{(\kappa_{2}u,\kappa_{2}v)}{(u,v)\in\rellift(F_{2})(R)}$;

\item $\rellift(F^{A})(R) = \set{(f,g)}{\allin{a}{A}{(f(a),g(a)) \in
   \rellift(F)(R)}}$;

\item $\rellift(\Pow)(R) =
  \set{(U,V)}{\allin{x}{U}{\exin{y}{V}{(x,y)\in R}} \mbox{ and }
    \allin{y}{V}{\exin{x}{U}{(x,y)\in R}}}$. \QED
\end{enumerate}
\end{lem}

\begin{exa}
\label{BisimEx}
We shall elaborate what it means that a coalgebra $c \colon Y
\rightarrow F(Y)$ also forms a coalgebra $c\colon R \rightarrow
\rellift(F)(R)$, for a relation $R\subseteq Y\times Y$. We shall do so
for three different instantiations of the functor $F$. This gives
uniform description of the same examples used in
Section~\ref{sec:bisapt}.
\begin{enumerate}
\item Let $F$ be the functor $F(X) = A\times X$ for streams, as in
  Definition~\ref{def.streams}. We can then describe the coalgebra $c
  \colon Y \rightarrow A\times Y$ as a pair $c = \tuple{h,t}$ for a
  ``head'' function $h\colon Y \rightarrow A$ and a ``tail'' function
  $t\colon Y \rightarrow Y$. The lifting is:
\[ \begin{array}{rcl}
\rellift(F)(R)
& = &
\set{((a_{1},y_{1}), (a_{2},y_{2}))}{a_{1} = a_{2} \mbox{ and }
   R(y_{1},y_{2})}
\\
& = &
\set{((a,y_{1}), (a,y_{2}))}{R(y_{1},y_{2})}.
\end{array} \]

\noindent Thus, $\tuple{h,t}$ being a map $R \rightarrow \rellift(F)(R)$
in $\Rel$ corresponds to the usual definition of bisimulation:
\[ R\big(x_{1},x_{2}\big)
\;\Longrightarrow\;
h(x_{1}) = h(x_{2}) \mbox{ and } R\big(t(x_{1}), t(x_{2})\big). \]

\item For deterministic automata one uses the functor $F(X) = X^{A}
  \times 2$, where $2 = \{0,1\}$, as in Definition~\ref{def.DFA}. A
  coalgebra is again a tuple $c = \tuple{c_{1},c_{2}} \colon Y
  \rightarrow Y^{A} \times 2$, with the following standard notation:
\[ \begin{array}{rclcrcl}
y \trans_{a} y'
& \Leftrightarrow &
c_{1}(y)(a) = y'
& \qquad\mbox{and}\qquad &
y\downarrow
& \Leftrightarrow &
c_{2}(y) = 0.
\end{array} \]

\noindent The notation $y\downarrow$ denotes that $y$ is a final
state. Having a coalgebra $\tuple{c_{1},c_{2}} \colon R \rightarrow
\rellift(F)(R)$ now means:
\[ R(x_{1},x_{2})
\;\Longrightarrow\;
\left\{\begin{array}{l}
x_{1}\downarrow \mbox{ iff } x_{2}\downarrow
\\
\quad\mbox{and}
\\
\all{a_{1},a_{2}}{\all{y_{1},y_{2}}{x_{1} \trans_{a_1} y_{1} \;\&\; 
   x_{2} \trans_{a_2} y_{2} \Rightarrow a_{1} = a_{2} \;\&\; R(y_{1},y_{2})}}.
\end{array}\right. \]

\item We now use a functor $F(X) = \Pow\big(X + X\times A\times
  X\big)$ and investigate the associated form of bisimulation. We show
  that it resembles the formulation that we have seen earlier for
  weaker forms of bisimulation.

So lets start with a coalgebra $c \colon Y \rightarrow \Pow\big(Y +
Y\times A\times Y\big)$ and write:
\[ \begin{array}{rclcrcl}
x \trans_{1} y
& \Leftrightarrow &
\kappa_{1}y\in c(x)
& \qquad\mbox{and}\qquad &
x \trans_{2} y \trans_{a} z
& \Leftrightarrow &
\kappa_{2}(y,a,z) \in c(x).
\end{array} \]

\noindent Here we regard $\trans_{1}$ and $\trans_{2}$ as two
different forms of silent steps. There is a coalgebra $c \colon R
\rightarrow \rellift(F)(R)$ when:
\[ R(x_{1},x_{2})
\;\Longrightarrow\;
\left\{\begin{array}{l}
x_{1} \trans_{1} y_{1} \Longrightarrow \ex{y_2}{x_{2} \trans_{1} y_{2}}, 
  \quad\mbox{and}
\\
x_{1} \trans_{2} y_{1} \trans_{a} z_{1} \Longrightarrow
  \ex{y_{2}, z_{2}}{x_{2} \trans_{2} y_{2} \trans_{a} z_{2}}
\\
\quad\mbox{and}
\\
x_{2} \trans_{1} y_{2} \Longrightarrow \ex{y_1}{x_{1} \trans_{1} y_{1}}, 
  \quad\mbox{and}
\\
x_{2} \trans_{2} y_{2} \trans_{a} z_{2} \Longrightarrow
  \ex{y_{1}, z_{1}}{x_{1} \trans_{2} y_{1} \trans_{a} z_{1}}.
\end{array}\right. \]

\noindent These clauses show that with this categorical way of
capturing silent steps (via $\trans_{1}$ and $\trans_{2}$) means that
the numbers of silent steps in both coordinates must be equal. This
differs in an important way from weak/branching bisimulation where a
single silent step on one coordinate can be mimicked by multiple (zero
or more) silent steps in the other coordinate. We refer to the
literature for different ways of bringing the two approaches closer
togeter~\cite{SokolovaVinkWoracek,BeoharKupper,Brengos15,BrengosMP15,GoncharovP14}.
\end{enumerate}
\end{exa}

The next result gives a concrete description of what is captured
abstractly in~\cite{HermidaJ98}.

\begin{prop}
\label{EqFunProp}
The equality functor $\Eq \colon \Set \rightarrow \Rel$ restricts to
$\Eq \colon \CoAlg(F) \rightarrow \CoAlg(\rellift(F))$, for each
functor $F$, and preserves final coalgebras. This implies that
bisimilar elements become equal when mapped to the final coalgebra.
\end{prop}

\begin{myproof}
Let $c \colon X \rightarrow F(X)$ be an arbitrary coalgebra. Applying
the equality functor to it yields a $\rellift(F)$-coalgebra:
\[ \xymatrix@C+1pc{
\Eq(X)\ar[r]^-{c} & \Eq(F(X)) = \rellift(F)(\Eq(X)).
} \]

\noindent Now let $\zeta \colon Z \stackrel{\cong}{\longrightarrow}
F(Z)$ be the final $F$-coalgebra. Let $c \colon R \rightarrow
\rellift(F)(R)$ be a coalgebra, so that $R$ is a bisimulation on $c
\colon X \rightarrow F(X)$. We write $X/R$ for the quotient of $X$
under (the least equivalence relation containing) $R$, with quotient
map $q_{R} \colon X \rightarrow X/R$. Then $q_{R} \colon R
\hookrightarrow \Eq(X/R)$ in $\Rel$, giving:
\[ \xymatrix@C+1pc{
R\ar[r]^-{c} & \rellift(F)(R)\ar[r]^-{F(q_{R})} &
   \rellift(F)\big(\Eq(X/R)\big) = \Eq(F(X/R)).
} \]

\noindent This means that there is a unique coalgebra $c/R \colon X/R
\rightarrow F(X/R)$ with $c/R \after q_{R} = F(q_{R}) \after c$. As a
result, the unique coalgebra homomorphism $g$ from $c$ to the final coalgebra
$\zeta$ factors as:
\[ \xymatrix{
F(X)\ar[rr]^-{F(q_{R})} & & F(X/R)\ar@{-->}[rr]^-{F(f)} & & F(Z)
\\
X\ar[u]^{c}\ar[rr]^-{q_R}\ar@/_3ex/[rrrr]_-{g} & & 
   X/R\ar[u]^{c/R}\ar@{-->}[rr]^-{f} & & Z\ar[u]_-{\cong}^-{\zeta}
} \]

\noindent Then, for $(x,x') \in R$, $q_{R}(x) = q_{R}(x')$, and
therefore $g(x) = g(x')$. This means that $g$ is a
$\rellift(F)$-coalgebra homomorphism $R \rightarrow \Eq(Z)$. By
finality of $\zeta \colon Z \stackrel{\cong}{\longrightarrow} F(Z)$
there can be at most one such homomorphism. \QED
\end{myproof}

We conclude this section with the following observation. As we have seen,
$R$ is a bisimulation for a coalgebra $c$ when:
\[ \begin{array}{rcl}
R
& \subseteq &
(c\times c)^{-1}\big(\rellift(F)(R)\big).
\end{array} \]

\noindent Thus, the bisimilarity $\bis^{c}$ --- that is, the greatest
bisimulation on $c$ --- can be obtained as the greatest post-fixed
point (final coalgebra) of the monotone operator $R \longmapsto
(c\times c)^{-1}\big(\rellift(F)(R)\big)$.

\section{Apartness}\label{ApartnesSec}

The opposite $\op{\Rel}$ of the category of relations contains
relations as objects with reversed arrows. We are going to use a
different category $\fop{\Rel}$ which is the ``fibred opposite'',
where the order relations in the fibres are reversed. This is an
instance of a more general
construction~\cite[Defn.~1.10.11]{Jacobs99}.

\begin{defi}
\label{FopDef}
The category $\fop{\Rel}$ has binary relations as objects. A morphism
$f\colon (R\subseteq X\times X) \rightarrow (S\subseteq Y\times Y)$ is
a function $f\colon X\rightarrow Y$ satisfying $R \supseteq (f\times
f)^{-1}(S)$. This means that $(f(x_{1}),f(x_{2})) \in S$ implies
$(x_{1},x_{2})\in R$. There is an obvious forgetful functor
$\fop{\Rel} \rightarrow \Set$, given by $(R\subseteq X\times X)
\mapsto X$.

\begin{enumerate}
\item We shall write $\neg \colon \Rel \rightarrow \fop{\Rel}$ for the
  negation functor, where $\neg R = \set{(x,x')}{(x,x')\not\in R}$.
  On morphisms we have $\neg(f) = f$.

\item We write $\nEq \coloneqq \neg \after \Eq \colon \Set
  \rightarrow \fop{\Rel}$ for the inequality functor, sending a set
  $X$ to $\nEq(X) = \set{(x,x')}{x\neq x'} \subseteq X\times X$.

\item For a functor $F\colon \Set \rightarrow \Set$ we define
  `opposite relation lifting' $\foprellift(F) \coloneqq \neg \after
  \rellift(F) \after \neg \colon \fop{\Rel} \rightarrow
  \fop{\Rel}$. Then we have a situation:
\[ \xymatrix@R-1pc@C-1pc{
& & \fop{\Rel}\ar[rrrrrr]^-{\foprellift(F)}\ar[ddl]|(0.52){\hole} 
   & & & & & & \fop{\Rel}\ar[ddl]
\\
\Rel\ar[urr]^{\neg}\ar[rrrrrr]^(0.6){\rellift(F)}\ar[dr] & & & & & &
   \Rel\ar[urr]_-{\neg}\ar[dr]
\\
& \Set\ar[rrrrrr]^-{F} & & & & & & \Set
} \]
\end{enumerate}
\end{defi}

It is easy to see that $\fop{\big(\fop{\Rel}\big)} = \Rel$ and that
$\neg$ is a functor $\Rel \rightarrow \fop{\Rel}$: if $R \subseteq
(f\times f)^{-1}(S)$ then:
\[ \begin{array}{rcccl}
\neg R
& \supseteq &
\neg(f\times f)^{-1}(S)
& = &
(f\times f)^{-1}(\neg S).
\end{array} \]

\noindent Moreover, inequality is a functor since for $f\colon X
\rightarrow Y$ one has $f(x) \neq f(x') \Longrightarrow x\neq x'$,
that is, $\nEq(X) \supseteq (f\times f)^{-1}(\nEq(Y))$. Similarly,
$\neg$ can be described as a functor $\neg \colon \fop{\Rel}
\rightarrow \Rel$. Clearly, the composite $\neg \after \neg \colon
\Rel \rightarrow \fop{\Rel} \rightarrow \Rel$ is the identity functor,
and similarly for $\neg \after \neg \colon \fop{\Rel} \rightarrow \Rel
\rightarrow \fop{\Rel}$.

Since
$\rellift(F)$ commutes with equality, see Lemma~\ref{RelliftEqLem},
the opposite relation lifting $\foprellift(F)$ commutes with
inequality:
\[ \begin{array}{rcl}
\foprellift(F)\big(\nEq(X)\big)
\hspace*{\arraycolsep}=\hspace*{\arraycolsep}
\neg\rellift(F)\big(\neg\neg\Eq(X)\big)
& = &
\neg\rellift(F)\big(\Eq(X)\big)
\\
& = &
\neg \Eq(F(X))
\hspace*{\arraycolsep}=\hspace*{\arraycolsep}
\nEq(F(X)).
\end{array} \]

We have the following analogue of Lemma~\ref{RelliftPolynomialLem} for
opposite relation lifting.

\begin{lem}
\label{FopRelliftPolynomialLem}
Relation lifting satisfies:
\begin{enumerate}
\item $\foprellift(\idmap) = \idmap$;

\item $\foprellift(A) = \nEq(A)$;

\item $\foprellift(F_{1}\times F_{2})(R) = 
   \set{(u,v)}{(\pi_{1}u,\pi_{1}v)\in\foprellift(F_{1})(R) \mbox{ or }
    (\pi_{2}u,\pi_{2}v)\in\foprellift(F_{1})(R)}$;

\item $\begin{array}[t]{rcl}\foprellift(F_{1} + F_{2})(R)
   & = &
   \set{(\kappa_{1}u, \kappa_{2}v)}{u\in F_{1}(X), v\in F_{2}(X)} \;\cup
\\
& &
   \set{(\kappa_{1}u, \kappa_{2}v)}{u\in F_{1}(X), v\in F_{2}(X)} \;\cup\
\\
& &
   \set{(\kappa_{1}u,\kappa_{2}v)}{(u,v)\in\foprellift(F_{1})(R)} \;\cup
\\
& &
   \set{(\kappa_{1}u,\kappa_{2}v)}{(u,v)\in\foprellift(F_{2})(R)}
\\
& = &
\{(a,b) \mid \all{u,v}{\begin{array}[t]{l}
   a = \kappa_{1}u, b = \kappa_{1}v \Rightarrow (u,v)\in\foprellift(F_{1})(R)
   \mbox{ and } \\
   a = \kappa_{2}u, b = \kappa_{2}v \Rightarrow (u,v)\in\foprellift(F_{2})(R)\}
   \end{array}}
\end{array}$

\item $\foprellift(F^{A})(R) = \set{(f,g)}{\exin{a}{A}{(f(a),g(a)) \in
   \foprellift(F)(R)}}$;

\item $\foprellift(\Pow)(R) =
  \set{(U,V)}{\exin{x}{U}{\allin{y}{V}{(x,y)\in R}} \mbox{ and }
    \exin{y}{V}{\allin{x}{U}{(x,y)\in R}}}$. \QED
\end{enumerate}
\end{lem}

\begin{exa}
\label{ApartEx}
We look at the analogue of Example~\ref{BisimEx}, using apartness
instead of bisimulation.
\begin{enumerate}
\item A relation $R\subseteq Y\times Y$ is an apartness relation
for the functor $F(X) = A\times X$ when there is a coalgebra
$\tuple{h,t} \colon R \rightarrow \foprellift(F)(R)$ in $\fop{\Rel}$.
This amounts to:
\[ \begin{array}{rclcrcl}
h(x_{1}) \neq h(x_{2})
& \Longrightarrow &
R\big(x_{1},x_{2}\big)
& \qquad\mbox{and}\qquad &
R\big(t(x_{1}),t(x_{2})\big)
& \Longrightarrow &
R\big(x_{1},x_{2}\big).
\end{array} \]

\noindent Apartness $\apt$ is the least relation $R$ for which these
two implications hold. In such a situation one commonly writes these
implications as rules:
\[ \begin{prooftree}
h(x_{1}) \neq h(x_{2})
\justifies
x_{1} \apt x_{2}
\end{prooftree}
\hspace*{8em}
\begin{prooftree}
t(x_{1}) \apt t(x_{2})
\justifies
x_{1} \apt x_{2}
\end{prooftree} \]

\noindent Alternatively, $x_{1} \apt x_{2}$ iff
$h\big(t^{n}(x_{1})\big) \neq h\big(t^{n}(x_{2})\big)$ for some
$n\in\mathbb{N}$.

\item For a deterministic automaton, $R$ is an apartness relation
  when:
\[ \left.\begin{array}{r}
x_{1}\downarrow \;\&\; \neg \big(x_{2}\downarrow\big)
\\
\mbox{or}\quad
\\
x_{2}\downarrow \;\&\; \neg \big(x_{1}\downarrow\big)
\\
\mbox{or}\quad
\\
\mbox{for some }a, \;
x_{1} \trans_{a} y_{1} \;\&\; x_{2} \trans_{a} y_{2} \;\&\; R\big(y_{1},y_{2}\big)
\end{array}\right\}
\;\Longrightarrow\;
R\big(x_{1},x_{2}\big). \]

\item For an apartness relation $R$ for the functor $F(X) = \Pow\big(X
  + X\times A\times X\big)$ there are many cases to distinghuish:
one has $R(x_{1},x_{2})$ if either:
\begin{itemize}
\item $x_{1} \trans_{1} y_{1}$, but there is no $y_{2}$ with $x_{2}
  \trans_{1} y_{2}$;

\item $x_{1} \trans_{1} y_{1}$ and $x_{2} \trans_{1} y_{2}$ with
  $R(y_{1}, y_{2})$;

\item $x_{1} \trans_{2} y_{1} \trans_{a_1} z_{1}$, but there are no
$y_{2},a_{2}, z_{2}$ with $x_{2} \trans_{2} y_{2} \trans_{a_2} z_{2}$;

\item $x_{1} \trans_{2} y_{1} \trans_{a} z_{1}$, but there are no
$y_{2},z_{2}$ with $x_{2} \trans_{2} y_{2} \trans_{a} z_{2}$;

\item $x_{1} \trans_{2} y_{1} \trans_{a} z_{1}$ and
$x_{2} \trans_{2} y_{2} \trans_{a} z_{2}$ with $R(y_{1},y_{2})$;

\item $x_{1} \trans_{2} y_{1} \trans_{a} z_{1}$ and
$x_{2} \trans_{2} y_{2} \trans_{a} z_{2}$ with $R(z_{1},z_{2})$,
\end{itemize}

\noindent and similarly for the six symmetric cases, starting with $x_{2}$.
\end{enumerate}
\end{exa}

\begin{defi}
\label{ApartnessDef}
Let $c \colon X \rightarrow F(X)$ be an arbitrary coalgebra.
\begin{enumerate}
\item An \emph{apartness relation} for $c$ is a relation $R$ on $X$
for which $c$ is a coalgebra $c\colon R \rightarrow \foprellift(F)(R)$.
This means that:
\[ \begin{array}{rclcrcl}
R 
& \supseteq &
(c\times c)^{-1}\big(\foprellift(F)(R)\big)
& \quad\mbox{or, equivalently}\quad &
\neg R
& \subseteq &
(c\times c)^{-1}\big(\rellift(F)(\neg R)\big).
\end{array} \]

\noindent Thus, $R$ is an apartness relation iff $\neg R$ is a
bisimulation relation.

\item \emph{Apartness} $\apt^{c}$ is the greatest apartness relation,
  w.r.t.\ the order $\supseteq$.
\end{enumerate}
\end{defi}

The category $\CoAlg(\foprellift(F))$ thus has apartness relations as
objects. In order to find the apartness relation on a coalgebra $c
\colon X \rightarrow F(X)$ we need to find the \emph{greatest
  post-fixed point} (final coalgebra) of the mapping $R \longmapsto
(c\times c)^{-1}\big(\rellift(F)(\neg R)\big)$ in the poset
$(\Pow(X\times X), \supseteq)$ with opposite order. A crucial
observation is that this is the \emph{least pre-fixed point} (initial
algebra) in $\Pow(X\times X)$ with usual inclusion order
$\subseteq$. Elements of such an initial algebra can typically be
constructed in a finite number of steps. This corresponds to the idea
that finding a difference in behaviour of coalgebra can be done in
finitely many steps --- although you may not know how many steps ---
whereas showing indistinguishability of behaviour involves all steps.

\begin{prop}
\label{nEqFunProp}
The inequality functor $\nEq \colon \Set \rightarrow \fop{\Rel}$
restricts to $\nEq \colon \CoAlg(F) \rightarrow
\CoAlg(\foprellift(F))$, for each functor $F$, and preserves final
coalgebras. This implies that elements that are apart become non-equal
(different) when mapped to the final coalgebra.
\end{prop}

\begin{myproof}
We assume a final $F$-coalgebra $\zeta \colon Z \rightarrow F(Z)$.
Let $c \colon R \rightarrow \foprellift(F)(R)$ describe an apartness
relation $R\subseteq X\times X$, for a coalgebra $c\colon X
\rightarrow F(X)$. The latter has a unique coalgebra homomorphism
$f\colon X \rightarrow Z$. Since $\neg R$ is a bisimulation,
Proposition~\ref{EqFunProp} says that we have a map $f\colon \neg R
\rightarrow \Eq(Z)$ in $\Rel$. By functoriality of $\neg$ we get
$f\colon R \rightarrow \neg\Eq(Z) = \nEq(Z)$ in $\fop{\Rel}$, as
required. \QED
\end{myproof}

\section{Conclusion and Further directions}
In this paper we have explored the notion of ``apartness'' from a
coalgebraic perspective, as the negation of bisimulation. We have
shown what this means concretely in the simple cases of streams and
deterministic automata and in the cases of weak and branching
bisimulation for labelled transition systems. We have also given a
general categorical treatment of apartness as the negation of
bisimulation. An important contribution of this view is that it yields
a logic for proving that two states are apart, proving that they are
not bisimilar. This applies to the general situation for coalgebras of
a polynomial functor, but also to the specific situation of weak and
branching bisimulation.

It would be interesting to see if apartness can be applied in the
analysis or description of algorithms for computing bisimulation, in
particular branching bisimulation
\cite{Jansenetal,GrooteVaandrager}. In existing algorithms, apartness
clearly plays a role, as they are described in terms of a collection
of ``blocks'' that is refined. States in different blocks are apart,
so these algorithms refine an apartness until in the limit, the finest
possible apartness is reached. Our approach is not intended to replace
bisimulation with apartness. In fact, we view a combined approach as
most promising, e.g.\ as we have in Lemma \ref{lem.soundrules-bis}. In
the period of submission and revision of the present paper, the
concept of apartness has already been picked up and applied to develop
a new approach to active automata learning \cite{Vaandrageretal2021},
which underlines the potential fruitfulness of the concept.

\section*{Acknowledgments}
  \noindent We thank David N.\ Jansen for the fruitful discussions and
  we thank the reviewers for their valuable feedback and suggestions
  for improvements.

\bibliographystyle{alpha}
\bibliography{references}
  
\newpage
\appendix
\section{Syntax and logic for bisimulation and apartness}
\label{sec:appendix}
As we have seen, the treatment of bisimulation and apartness in Section \ref{sec:bisapt} can be generalized using category theory. We could also have chosen for a purely logical-syntactic approach, which we briefly summarize here. As a matter of fact, the treatment in Sections \ref{CoalgLiftSec} and \ref{ApartnesSec} can be seen as a semantics of the syntax we introduce here.

We start with a description of the systems we want to deal with, which
are basically a slight generalization of polynomial functors. A
coalgebra has a carrier and a destructor operation. We describe the
``types'' of the possible range of a destructor.

\begin{defi} \label{D:destructors}
Let a set $\At$ of symbols fot atomic types be given and let $X$ be a
special symbol (a type variable) not in $\At$.
\[\Ty ::= \At \mid X \mid 1 \mid \Ty + \Ty \mid \Ty \times \Ty\mid
\Ty^{\At} \mid 2^{\Ty}.\]
A {\em destructor signature\/} is a pair $(d,\sigma)$ with
\[d: X \rightarrow \sigma.\]
A {\em coalgebra for the destructor signature $(d,\sigma)$\/} is a $K
\in \Set$ and $h: K \rightarrow F(K)$ (the {\em destructor}), where
$F$ is the functor on $\Set$ defined from $X\mapsto \sigma(X)$ in the
obvious way.
\end{defi}

We have seen examples for this definition in \ref{def.DFA},
\ref{def.streams}, \ref{BisimEx} and \ref{ApartEx}.

\begin{defi}\label{D:bisim}
Let $K$ be a coalgebra for the signature of Definition
\ref{D:destructors} with destructor $h: K \rightarrow F(K)$.
\begin{enumerate}
\item A relation $R \subseteq K \times K$ is called a {\em bisimulation} (or
a {\em coinduction assumption}) on $K$ if the following holds for all
$x, y \in K$.
\[ R(x,y) \Rightarrow h(x) \bisim{F(K)}{R} h(y).\]
Here, the relation $\bisim{B}{R}$ is defined by induction on the
structure of $B$ as follows.
\begin{itemize}
\item $B=D\in \At$, then $x\bisim{B}{R}y := \Eq_D(x,y)$, equality on $D$.
\item $B=X$, then $x\bisim{B}{R}y := R(x,y)$.
\item $B=1$, then $x\bisim{B}{R}y$ is true.
\item $B=B_1 + B_2$, then $x\bisim{B}{R}y :=$\\
  $\exists u, v (x = \inl(u)
  \wedge y = \inl(v) \wedge u \bisim{B_1}{R}v) \vee (x = \inr(u)
  \wedge y = \inr(v) \wedge u \bisim{B_2}{R}v)$.
\item $B=B_1 \times B_2$, then $x\bisim{B}{R}y :=$\\
  $\exists
  u_1,u_2,v_1,v_2 (x = \pair{u_1}{u_2} \wedge y = \pair{v_1}{v_2}
  \wedge u_1 \bisim{B_1}{R}v_1 \wedge u_2 \bisim{B_2}{R}v_2)$.
\item $B=B_1^D$ with $D\in \At$, then $x\bisim{B}{R}y := \forall d\in
  D (x(d) \bisim{B_1}{R} y(d))$.
\item $B=2^{B_1}$, then $x\bisim{B}{R}y := (\forall u(x(u) =1 \implies
  \exists v (y(v) =1 \wedge u \bisim{B_1}{R} v)) \wedge (\forall
  v(y(v) =1 \implies \exists u (x(u) =1 \wedge u \bisim{B_1}{R} v))$.
\end{itemize}
\item Two elements $x,y \in K$ are called {\em bisimilar}, notation $x\bis y$, if there exists a bisimulation $R$ for which $R(x,y)$ holds. Thus:
\[\bis \;=\, \bigcup\{R \subseteq K\times K \mid R \mbox{ is a bisimulation} \}.\]
\end{enumerate}
\end{defi}

The condition that $R\subseteq K\times K$ should satisfy in order to
be a bisimulation can be viewed as a {\em derivation rule}.  This
gives an alternative way to say that $R$ is a bisimulation.

\begin{rem}\label{R:bisim} 
The relation $R\subseteq K\times K$ is a bisimulation if the following
derivation rules hold for $R$.
\[\begin{prooftree}
R(x,y)
\justifies 
h(x) \bisim{F(K)}{R} h(y)
\end{prooftree}\]
\end{rem}

\begin{lem}
The relation $\bis$ is itself a bisimulation relation, and therefore
$\bis$ is the largest bisimulation relation.
\end{lem}

\begin{proof}
This is a standard fact about bisimulations that is easily verified
for this more general setting. The crucial propertyis that if $x
\bisim{B}{R} y$ for some bisimulation $R$, then $x \bisim{B}{\bis} y$.
\end{proof}

Similar to the definition of bisimulation (Definition \ref{D:bisim}),
we can give a general definition of apartness for a coalgebra of a
destructor signature.

\begin{defi}\label{D:apart}
Let again $K$ be a coalgebra for the signature of Definition
\ref{D:destructors} with destructor $h: K \rightarrow F(K)$.
\begin{enumerate}
\item A relation $Q \subseteq K \times K$ is called an {\em apartness} on $K$ if the following holds for all
$x, y \in K$.
\[ h(x) \apart{F(K)}{Q} h(y)  \Rightarrow Q(x,y).\]
Here, the relation $\apart{B}{Q}$ is defined by induction on the
structure of $B$ as follows.
\begin{itemize}
\item $B=D\in \At$, then $x\apart{B}{Q}y := \neg\Eq_D(x,y)$, in-equality on $D$.
\item $B=X$, then $x\apart{B}{Q}y := Q(x,y)$.
\item $B=1$, then $x\apart{B}{Q}y$ is false.
\item $B=B_1 + B_2$, then $x\apart{B}{Q}y :=$\\
  $\forall u, v ((x = \inl(u)
  \wedge y = \inl(v)) \Rightarrow u \apart{B_1}{Q}v) \wedge ((x = \inr(u)
  \wedge y = \inr(v)) \Rightarrow u \apart{B_2}{Q}v)$.
\item $B=B_1 \times B_2$, then $x\apart{B}{R}y :=$\\
  $\forall
  u_1,u_2,v_1,v_2 (x = \pair{u_1}{u_2} \wedge y = \pair{v_1}{v_2})
  \Rightarrow u_1 \apart{B_1}{Q}v_1 \vee u_2 \apart{B_2}{Q}v_2)$.
\item $B=B_1^D$ with $D\in \At$, then $x\apart{B}{R}y := \exists d\in
  D (x(d) \apart{B_1}{R} y(d))$.
\item $B=2^{B_1}$, then $x\apart{B}{R}y := (\exists u(x(u) =1 \wedge
  \forall v (y(v) =1 \implies u \apart{B_1}{R} v)) \vee$\\
  $(\exists
  v(y(v) =1 \wedge \forall u (x(u) =1 \implies u \apart{B_1}{R} v))$.
\end{itemize}
\item Two elements $x,y \in K$ are called {\em apart}, notation $x\apt y$, if for all apartness relations $Q$ we have $Q(x,y)$. Thus:
\[\apt \;=\, \bigcap\{Q \subseteq K\times K \mid Q \mbox{ is an apartness} \}.\]
\end{enumerate}
\end{defi}

Just as for bisimulation, we can phrase the property that a relation
should satisfy in order to be an apartness in the form of aderivation
rule. 

\begin{rem}\label{R:apart}
A relation $Q\subseteq K\times K$ is an apartness if the
following derivation rule holds for $Q$.
\[\begin{prooftree}
h(x) \apart{F(K)}{Q} h(y)
\justifies 
Q(x,y)
\end{prooftree}\]
\end{rem}

\begin{lem}
The relation $\apt$ is itself an apartness relation, and therefore
$\apt$ is the smallest apartness relation.
\end{lem}

\begin{proof}
That the rule of Remark \ref{R:apart} is sound for $\apt$ is easily
verified by observing that, if $x \apart{B}{\apt}y$, then $x
\apart{B}{Q}y$ for any apartness $Q$.
\end{proof}

As $x\apt y$ is the smallest apartness relation, the rule of Remark
\ref{R:apart} gives a complete derivation system for proving $x\apt y$
for $x,y\in K$.

We can summarize the relations between apartness and bisimulation as
follows. This is the general statement of Lemma
\ref{lem.bisimandaprt}. The proof is simply checking all the
properties.

\begin{lem}
  \begin{enumerate}
  \item $R$ is a bisimulation if and only if $\neg R$ is an apartness.
  \item The relation $\bis$ is the union of all bisimulations, $\bis
    \;=\, \bigcup\{R \mid R \mbox{ is a bisimulation}\}$, and it is
    itself a bisimulation.
  \item The relation $\apt$ is the intersection of all apartness
    relations, $\apt \;=\, \bigcap\{Q \mid Q \mbox{ is an
      apartness}\}$, and it is itself an apartness.
  \item $\bis \;= \neg \apt$.
  \end{enumerate}
\end{lem}
\end{document}

\section{}
  Here is a check-list to be completed before submitting the paper to
  LMCS:
\begin{itemize}[label=$\triangleright$]
\item your submission uses the latest version of lmcs.cls
\item the text of your submission is contained in a single file,
  except for macros and graphics
\item your graphics use only one format 
\item you have employed the Journal's original proclamation environments,
  or suitable extensions thereof 
\item you have loaded the hyperref package
\item you have \emph{not} loaded the times package
\item you have not routinely adjusted vertical spacing manually by issuing
  \texttt{\textbackslash vspace} or \texttt{\textbackslash vskip} commands
\item you have used the command \texttt{\textbackslash sloppy} only
  locally and in emergency cases
\item your displayed equations use the
  \texttt{\textbackslash[\dots\textbackslash]} construct
\item your abstract only contains as few math-expressions as possible and no
  references 
\end{itemize}

  This listing also shows how to override the default bullet $\bullet$
  of the \texttt{itemize}-envronment by a different symbol, in this
  case \texttt{\textbackslash triangleright}.
\end{document}